\documentclass[a4paper,UKenglish,cleveref, autoref, thm-restate, numberwithinsect ]{lipics-v2021}

\title{Out-of-Order Membership in Regular Languages}

\usepackage[bibliography=common]{apxproof}

\author{Antoine Amarilli}{Univ. Lille, INRIA, CNRS, Centrale Lille, UMR 9189
CRIStAL, F-59000 Lille, France}{antoine.a.amarilli@inria.fr}{https://orcid.org/0000-0002-7977-4441}{}
\author{Sebastien Labbe}{Univ. Lille, CNRS, INRIA, Centrale Lille, UMR 9189 CRIStAL, F-59000 Lille, France}{sebastien.labbe@pm.me}{https://orcid.org/0009-0007-1790-7820}{}
\author{Charles Paperman}{Univ. Lille, CNRS, INRIA, Centrale Lille, UMR 9189 CRIStAL, F-59000 Lille, France}{charles.paperman@univ-lille.fr}{https://orcid.org/0000-0002-6658-5238}{ANR-24-CE25-2874}

\authorrunning{A.\ Amarilli, S.\ Labbe, C.\ Paperman} 

\Copyright{A.\ Amarilli, S.\ Labbe, C.\ Paperman} 



\category{Track B: Automata, Logic, Semantics, and Theory of Programming}

\relatedversion{} 

\nolinenumbers 

\hypersetup{hypertexnames=false}

\EventEditors{Sayan Bhattacharya, Danupon Nanongkai, Michael Benedikt, and Gabriele Puppis}
\EventNoEds{4}
\EventLongTitle{53rd International Colloquium on Automata, Languages, and Programming (ICALP 2026)}
\EventShortTitle{ICALP 2026}
\EventAcronym{ICALP}
\EventYear{2026}
\EventDate{July 7--10, 2026}
\EventLocation{Royal Holloway, University of London, Egham, United Kingdom}
\EventLogo{}
\SeriesVolume{374}
\ArticleNo{171}

\newtheoremrep{theorem}{Theorem}[section]
\newtheoremrep{proposition}[theorem]{Proposition}
\newtheoremrep{claim}[theorem]{Claim}
\newtheoremrep{corollary}[theorem]{Corollary}
\newtheoremrep{lemma}[theorem]{Lemma}
\newtheoremrep{example}[theorem]{Example}
\newtheoremrep{definition}[theorem]{Definition}

\usepackage{amsmath,amssymb,amsthm,mathtools}

\usepackage{stmaryrd}

\date{}

\newcommand{\deff}[1]{\emph{#1}}

\newcommand{\ept}{\, \_ \,}

\newcommand{\FF}{\mathrm{F}}
\newcommand{\LL}{\mathrm{L}}
\newcommand{\FL}{\mathrm{FL}}

\newcommand{\dom}{\mathrm{dom}}

\usepackage{booktabs}

\usepackage{tikz}
\usetikzlibrary{arrows.meta,positioning}
\usetikzlibrary{calc}

\begin{document}

\maketitle

\begin{abstract}
We introduce the task of \emph{out-of-order membership} to a formal language
\(L\), where the letters of a word \(w\) are revealed one by one in an
arbitrary order. The length \(|w|\) is known in advance, but the
content of \(w\) is streamed as pairs \((i, w[i])\), received exactly once for
each position \(i\), in arbitrary order. We study efficient algorithms
for this task when \(L\) is regular, seeking tight complexity bounds as
a function of \(|w|\) for a fixed target language. Most of our results
apply to an algebraically defined variant dubbed \emph{out-of-order
evaluation}: this problem is defined for a
fixed finite monoid or semigroup \(S\), and our goal is to compute
the ordered product of the streamed elements of~\(w\).

We show that, for any fixed regular language or finite semigroup, both problems can be solved
in constant time per streamed symbol and in linear space. However, the
precise space complexity
strongly depends on the algebraic structure of the
target language or evaluation semigroup. Our main contributions are
therefore to show (deterministic) space complexity characterizations,
which we do for out-of-order evaluation of monoids and semigroups.

For monoids, we establish a trichotomy: the space complexity is either
\(\Theta(1)\), \(\Theta(\log n)\), or \(\Theta(n)\), where \(n = |w|\).
More specifically, the problem admits a constant-space solution for commutative monoids,
while all non-commutative monoids require \(\Omega(\log n)\) space. We
further identify a class of monoids admitting an \(O(\log n)\)-space
algorithm, and show that all remaining monoids require \(\Omega(n)\)
space.

For general semigroups, the situation is more intricate. We characterize a
class of semigroups admitting constant-space algorithms for
out-of-order evaluation, and show that semigroups outside this class
require at least \(\Omega(\log n)\) space. At the same time, we exhibit
semigroups for which specialized techniques yield intermediate bounds
such as an \(O(\sqrt{n})\)-space algorithm, suggesting that the landscape may be
richer and less well-behaved than for the monoid setting.
\end{abstract}


\section{Introduction}
\label{sec:introduction}
The class of regular languages can be seen as a robust formalism to define the
properties of words that can be tested efficiently, i.e., using constant time per symbol
and constant memory overall. These bounds can be achieved by feeding the input
word to an automaton: they hold when the word is stored in an array and read from
left-to-right, and also when the word is streamed character by character.

However, it is not always realistic to assume that words can be processed from
left-to-right: in some contexts the word is
revealed character by character in an out-of-order fashion. One example is
distributed processing frameworks (e.g., MapReduce), in the spirit of
out-of-order
data processing~\cite{akidau2015dataflow}: if we want to
determine some property of a sequence of
data items processed concurrently by multiple workers, then
the value of each item will typically arrive as each worker finishes, in an
unpredictable order (see, e.g., \cite{ananthanarayanan2010reining}). Another
example is networking: large data typically travels using packet switching, and
the packets typically arrive in an out-of-order fashion. There, a line of work has investigated
how to efficiently detect patterns in such a stream of out-of-order packets, without
reassembly~\cite{varghese2006detecting}, with algorithms adapted from suffix
trees~\cite{chen2011ac} or automata~\cite{zhang2003space,yu2016o3fa}.

Such practical scenarios lead us to consider the task of
\emph{out-of-order membership} to a regular language, which we think is also
a natural theoretical question. In this problem, we fix a regular language~$L$,
we consider a word $w$ whose length
$|w|$ is known in advance, and we are streamed the contents of~$w$ in an
out-of-order fashion, i.e., as pairs $(i,w[i])$ giving us the letter at
position~$i$ of~$w$. Each position is streamed exactly once, and there is no
guarantee on the order (i.e., it is adversarial). The algorithm must 
correctly determine, at the end of the stream, whether $w \in L$ or $w \notin
L$, while
minimizing the space usage and the time spent processing each
character. In particular, in contrast with left-to-right streaming, we will see
that it may no longer be possible to decide membership with a constant state
complexity, i.e., the memory usage may depend on the length of the word.
To our knowledge this natural question
of the complexity of out-of-order membership has not yet been studied, though it
relates to many lines of work on membership variants,
presented next.

\subparagraph*{Related work.}
One first related topic is the \emph{query complexity} of regular
languages~$L$, i.e.,
the asymptotic number of (interactively chosen) characters that need to be
inspected in an input word $w$ to decide whether $w \in L$.
This problem requires $\Omega(n)$ on words of length~$n$ for almost all regular
languages~$L$ (and in $O(1)$ for a narrow class); but it admits a more interesting
classification in the quantum setting~\cite{aaronson2019quantum}. This question has been further
extended to \emph{property
testing}~\cite{alon2001regular} where we only need to
discriminate between words in~$L$ and words that are far from $L$
according to a certain distance. 
In this context, a trichotomy on the query complexity of
property testing has been achieved~\cite{bathie2025trichotomy}.
However, all of these problems differ from
out-of-order membership: they assume that we can freely choose positions
to inspect, whereas in out-of-order membership the streaming order is arbitrary. They also
focus on the number of queries needed, whereas with out-of-order membership we
receive all positions and focus on the time and space complexity.

Another related line of work is \emph{sliding window membership}, where we want
to test membership to a fixed language~$L$ for a sliding window over a stream.
Works on this problem have investigated
the complexity of various languages~$L$, first in terms of time per
character~\cite{ganardi2022low}, and then in terms of memory usage, with
again a recent trichotomy on space complexity for regular languages
in several settings~\cite{Ganardi2024Sliding}.
Other works extend sliding windows by considering possible arrival of
out-of-order elements,
in particular for
aggregation~\cite{traub2018scotty,tangwongsan2023out,bou20241}.
However, one important general difference of the sliding window setting
is that we must 
maintain membership to~$L$ for the various states of the sliding window across the
stream, with letters that leave
the sliding window when they expire. By contrast, in our context, we only test
membership of a single word, and revealed positions never expire; but the arrival
order is arbitrary and not sequential. Hence, for instance, the language $a \Sigma^*$ admits
constant space out-of-order membership (by storing the first character)
but requires linear space in the sliding
window setting~\cite{Ganardi2024Sliding}, and conversely one can show that out-of-order membership to the
language $\Sigma^* aa \Sigma^*$ requires linear space (see \cref{mon_notricom_lin})
but it can be tested in
logarithmic space in the sliding window setting~\cite{Ganardi2024Sliding}.

Another related setting is \emph{dynamic membership}, where we must test 
whether a word $w$ belongs to a regular language~$L$ and maintain this information under
updates to~$w$ (not just the push-right and pop-left updates of
sliding windows). In particular, a conditional trichotomy
has been shown under substitution updates~\cite{Amarilli2021Dynamic}, i.e., where we receive updates $(a,i)$ that change the $i$-th letter of
the word to the letter~$a$. Out-of-order membership is analogous to a special 
case of this setting where each position is set exactly once and never changed
afterwards.
However, \cite{Amarilli2021Dynamic} always assumes to
have linear memory to store the current state of~$w$ -- their trichotomy
is on time complexity and not space complexity.
In fact, unlike their setting,
we will show that all regular languages admit an out-of-order membership
algorithm in $O(1)$
time per character.

Last, the setting closest to ours is the \emph{communication
complexity} of regular language membership, studied in particular by Tesson and
Thérien~\cite{tesson} and then by Ada~\cite{ada2010non}. In this setting, we wish
again to determine whether a word $w$ belongs to a regular language~$L$, but
$w$ is
split across multiple players and we wish to minimize how much information
they need to exchange to determine whether $w \in L$.
Many communication complexity models are quite different from out-of-order
membership, e.g., the
multiparty setting of~\cite{RaymondTessonTherienCC}, along with
randomized settings~\cite{tesson} and nondeterministic
settings~\cite{ada2010non}. The most relevant model 
is the \emph{two-player setting} of~\cite{tesson}: the word is split
between Alice (for odd positions) and Bob (for even positions), and they can
exchange information in arbitrarily many rounds to determine membership. 
Communication lower bounds in this model imply space
complexity lower bounds in our setting: if we are first streamed the even positions,
then the information that we remember gives a one-way protocol from Alice to
Bob.
%
However, upper bounds in their setting
do not necessarily translate to ours, because our algorithms do not obey a fixed
partition.
As we show in the special case of languages featuring a neutral
letter (i.e., the case of monoids), the complexity of out-of-order membership happens to
coincide with the communication complexity for another model studied in~\cite{tesson}, namely, the
\emph{simultaneous setting} in which Alice and Bob must each send information to a referee. This is surprising because, as far as we understand, there is no
reduction from their model to ours or vice-versa. Nevertheless, as the
tractability boundary is the same, our algebraic proofs in this specific setting
relate to theirs, as we will
explain. In any case, all results in~\cite{tesson} crucially assume the presence of a
neutral letter (i.e., some of the letters of Alice and/or Bob may be
$\epsilon$):
the issue is further discussed in \cite[Section~5.2]{tesson1998algebraic}, and 
it has significant impact in our setting as well, as we will see when presenting
our contributions below.

\subparagraph*{Contributions.}
In this paper, we introduce the out-of-order membership problem for regular
languages, and give first results on its complexity, including a complete
space trichotomy classification for the case of monoids and a
characterization of constant space for semigroups. 
We present these contributions in more detail below along with the structure of the paper.

We introduce formally the problem and our model in \cref{sec:problem}, and
give some first complexity results. Specifically, in the unit-cost RAM model
with logarithmic word size,
we show that every fixed regular language admits an
out-of-order membership algorithm with constant time per streamed character and
linear memory usage (\cref{word_upper_n}). For this reason we only focus on space complexity in the
remainder of the paper. We close the section by adapting the classical
technique of \emph{fooling sets} for unconditional lower bounds, which we use
throughout the paper; we exemplify it to show a linear space lower bound for
$\Sigma^* aa \Sigma^*$ (\cref{mon_notricom_lin}).
By contrast, we show that better bounds are possible for some languages,
namely, commutative languages admit constant space out-of-order membership
(\cref{com_const}).

As classifying regular languages 
turns out to be challenging, we focus in the rest of the paper on
algebraically defined variants of the problem.
In \cref{sec:monoid}, we introduce the simplest such problem:
the \emph{out-of-order evaluation problem} for a monoid~$M$, where we are streamed a
word $w$ of monoid elements in an
out-of-order fashion and must compute the product of~$w$ in~$M$, while minimizing the space usage.
The point of out-of-order evaluation is that it can provide algorithms for
out-of-order membership to languages (via their syntactic monoid),
and it intuitively amounts to assuming that languages have been closed under the
addition of a neutral letter, simplifying the classification.
Further, we show that the complexity regimes of out-of-order monoid evaluation
form varieties of monoids, so the problem is amenable to algebraic tools.
We accordingly study out-of-order monoid evaluation and
prove (\cref{thm:monoid}) 
that monoids admit three space complexity regimes, with matching upper and
lower bounds: constant space for
commutative monoids, logarithmic space for monoids in a class dubbed $\mathbf{FL} \lor
\mathbf{Com}$, and linear space for all other monoids. The variety $\mathbf{FL} \lor
\mathbf{Com}$ intuitively extends commutative monoids with the ability to
remember a constant number of leftmost and rightmost occurrences of each
element. It turns out that the join variety $\mathbf{FL} \lor \mathbf{Com}$ is the same 
as the variety $\mathbf{W}$ introduced in~\cite{tesson} (with a less informative characterization).
Thus, the complexity of deterministic \emph{simultaneous} communication
complexity according to~\cite{tesson}, in the specific case of monoids, happens to
coincide with the space complexity of our out-of-order evaluation problem
-- even though neither problem seems to capture the other.
Our classification is summarized in the first column of \cref{all_table}, with
examples (as syntactic monoids of languages) in \cref{all_ex_table}.

We then turn in \cref{sec:semigroup} to the more general problem of out-of-order
evaluation for semigroups, i.e., monoids but without the requirement of having a neutral element.
The problem turns out to be much more challenging in the semigroup setting, and
thus we 
only focus on the constant space
regime, which we are able to characterize (\cref{low_log_sem}). Namely, we show a
constant space algorithm for semigroups in $\mathbf{Li} \lor \mathbf{Com}$,
i.e., intuitively extending commutative semigroups with the ability to test
constant-length prefixes and suffixes of the word (see \cref{all_ex_table} for
an example).
We show a logarithmic lower bound for all other semigroups.



In \cref{sec:horror}, we show that classifying the higher complexity regimes appears much more 
intricate. 
For convenience we present results on the out-of-order membership problem, but
the difficulties that we highlight 
already appear for 
out-of-order
semigroup evaluation (for space complexities beyond the constant regime).
We first
give 
ad hoc
logarithmic upper bounds for $a^* b^* a^*$ 
(\cref{prp:aba})
and for $a^* b^* a^* b^* a^*$ 
(\cref{prp:ababa})
using
completely different techniques.
We then show 
(\cref{prp:ababab})
an unexpected sublinear $O(\sqrt{n})$ space
upper bound for 
$a^* b^* a^* b^* a^*
b^*$ for which we have no better complexity upper bound and no superlogarithmic
lower bound.
See \cref{all_table} for the overall complexity bounds achieved for
out-of-order evaluation, and
\cref{all_ex_table} for examples.

We conclude and propose further research directions in \cref{sec:conc}. For lack
of space, most detailed proofs are deferred to 
the appendix.

\begin{table}[t]
	\centering
\begin{tabular}{cccc}
\toprule
\textbf{Space} & \textbf{Monoids} & \textbf{Semigroups} \\ 
\midrule
$O(1)$ & $\mathbf{Com}$ & $\mathbf{Li} \lor \mathbf{Com}$ \\ 
$O(\log n)$ & $\mathbf{FL} \lor \mathbf{Com}$ 
& ? \\ 
$O(n)$ & $\text{All monoids}$ 
& $\text{All semigroups}$ \\ 
\bottomrule
\end{tabular}
\caption{Out-of-order evaluation space complexity. The results of the first column are in \cref{thm:monoid}, the
$O(1)$ result for semigroups is in \cref{low_log_sem}.
	}
\label{all_table}
\end{table}


\begin{table}[t]
	\centering
        \begin{tabular}{c@{~~~}c@{~~~}c@{~~~}c}
\toprule
\textbf{Space} & \textbf{Monoids} & \textbf{Semigroups} & \textbf{Languages} \\
\midrule
	$\Theta(1)$ & $M((aa)^*)$ Ex.~\ref{ex:parity} & $S(a\Sigma^*b)$
        Prp.~\ref{ex:asigmab} & $(ab)^*$ Ex.~\ref{ex:abstar}\\
	$\Theta(\log n)$ & $M(ab)$ Ex.~\ref{ex:ab} & $S(a^*bc^*)$
		$S(a^*bba^*)$ Ex.~\ref{ex:sg_lower_bound}-\ref{ex:sg_lower_bound2} & $a^*b^*a^*$ Prp.~\ref{prp:aba} \\
	$O(\sqrt{n})$ &  &  & $a^*b^*a^*b^*a^*b^*$ Prp.~\ref{prp:ababa} \\ 
	$\Theta(n)$ & $M(a^* b b a^*)$ Ex.~\ref{ex:abba}& $S((ab)^*)$
        Prp.~\ref{ex:sg_lower_bound_ab}& $\Sigma^* aa \Sigma^*$ 
        Prp.~\ref{mon_notricom_lin}\\
        \bottomrule
\end{tabular}
\caption{Out-of-order evaluation/membership problem space complexity examples.
  The notations $M(\cdot)$ and $S(\cdot)$ respectively denote the syntactic
  monoid and syntactic semigroup of the indicated languages.}
\label{all_ex_table}
\end{table}

\section{Problem Statement and First Results}
\label{sec:problem}
In this section, we formally define the out-of-order membership problem, and
give a constant-time and linear-space upper bound for this problem on arbitrary
fixed regular languages (\cref{word_upper_n}).
Then we formally define fooling sets and give a
linear-space lower bound for the language $\Sigma^* aa \Sigma^*$
(\cref{mon_notricom_lin}).
We close by an easy constant-space upper bound for commutative languages 
(\cref{com_const}) which will be reused throughout the paper.

\subparagraph*{Problem statement.}
Let $\Sigma$ be an alphabet, let $n>0$ be a word length,
let $w = w[1] \cdots w[n]$ be a word over $\Sigma$ of length $n$, 
and let $\pi_n$ be an \emph{$n$-permutation},
i.e., a permutation of $\{1, \ldots, n\}$.
The \deff{out-of-order stream} of the word $w$ defined by~$\pi_n$
is the finite sequence $R$ of length $n$ whose successive elements are $(w[\pi(i)],
\pi(i), n)$ for $i = 1, \ldots, n$.
Intuitively, we receive the letters of~$w$ together with their position in the
order $\pi(1), \ldots, \pi(n)$. Note that the length $n$ of the word is provided
as the third component of each element: this is because we assume 
throughout the work that the length $n$ 
is always known and we do not want to
account for it when measuring space complexity.

We only consider regular languages in this paper, and we omit their definition.
A letter $a \in \Sigma$ is \emph{neutral} for a language $L \subseteq \Sigma^*$
if for any $s, t \in \Sigma^*$ we have $s a t \in L$ if and only if $s t \in L$.
For a fixed regular language $L$ over~$\Sigma$,
the \deff{out-of-order membership problem} for~$L$ is the following task:
letting $w$ be a word of length~$n$ over~$\Sigma$,
letting $\pi_n$ be an~$n$-permutation, and
letting $R$ be the out-of-order stream on~$w$ defined by~$\pi_n$,
we receive the elements of~$R$ one after the other and must decide at the end
whether $w \in L$ or not.

Note that the standard membership test to a fixed language~$L$ for a word $w$ streamed from
left-to-right corresponds to the case where $\pi_n$ is 
the
identity permutation. For left-to-right membership testing, we can simply use an
automaton for~$L$ and achieve constant space and constant
time per character; but for out-of-order membership the challenge is that we must support arbitrary
permutations. Further note that we
can always solve the problem by storing the letters in an array $T$ as they
arrive,
i.e., whenever we receive $(a, i, n)$ we set $T[i] \coloneq a$ and we test at
the end whether the word stored in~$T$ belongs to~$L$.
This gives a naive algorithm which uses
linear space and runs in constant worst-case complexity per streamed character
except that it spends linear time after the last character. Let us now present a more
efficient algorithm that avoids this drawback.



\subparagraph*{Time complexity upper bound.}
We first study efficient algorithms for out-of-order membership to fixed regular
languages in terms of time complexity per streamed character.
For this result on time complexity, and only for this result, we need to be
precise about the computational model that we assume; subsequent results will
only focus on space complexity and will apply no matter which computations are
allowed.
We work in the standard
unit-cost RAM model with logarithmic word
size~\cite{FredmanW90,grandjean2023which}.
Memory consists of cells of size $w = \Theta(\log n)$ bits, and reading, writing, or performing standard operations on these cells takes constant time. 
Our set of allowed operations includes addition, subtraction, increment,
comparisons, integer division with remainder, shifts, bitwise Boolean operations,
and computing the most significant bit.
We will also use the standard lazy array evaluation
technique to assume when allocating arrays that they can be initialized in
constant time~\cite{grandjean2023which}.
The algorithm receives the successive triples $(a, i, n)$ of the
out-of-order stream, with the values being stored in three $\Theta(\log n)$-bit cells. Here, $a$ uses constant bits (as the alphabet is fixed), while $i$ and $n$ are integers represented on $\Theta(\log n)$ bits.


In the computational model defined above, we can show that all regular languages admit an algorithm
with $O(1)$ complexity per letter and a linear space usage of $O(n)$ bits
(i.e., amounting to $O(n/\log n)$ memory cells).

\begin{theoremrep}
\label{word_upper_n}
For any fixed regular language $L$,
there is an algorithm for
out-of-order membership to~$L$ with \(O(1)\) time per streamed character and
a space complexity of \(O(n)\) bits.
\end{theoremrep}

\begin{proofsketch}
  Let $M$ be the syntactic monoid of~$L$ (see
        \cref{sec:monoid} for the formal definitions). We divide the input word into blocks
  of length $\lceil \log n\rceil$, so that each block can be represented in a
  constant number of RAM cells. The algorithm has two layers. The macro layer
  treats complete blocks as atomic symbols: for every maximal contiguous
  interval of complete blocks, it stores the interval length at both endpoints
  and the product in~$M$ of the block values at the left endpoint. The micro
  layer maintains analogous information inside each block, using one
  bit vector to record which positions of the block have already been streamed
  and a packed vector of monoid elements to store the values in~$M$ of maximal
  contiguous streamed intervals (which correspond to maximal contiguous
  intervals of 1's in the bit vector), each value being stored at the left
  endpoint of its interval.

  When a letter is streamed, we first update its block in the micro layer.
  To do this, we must compute the maximal contiguous interval of 1's to
  which the streamed letter belongs. This is accomplished by creating a
  singleton interval of 1's in the bit vector, and then trying to merge it with the interval
  of 1's immediately to its right and the interval immediately to its left.
  If there is an interval to the right, then merging with this interval is
  immediate because the monoid value of the interval is stored at its left
  endpoint (i.e., immediately to the right of the new position).
  For the merge with the left interval, we must find the left endpoint of that
  interval, which we do in constant time by shifting
  the bit vector of the block and applying the most-significant-bit operation.
  Hence the micro update takes constant time, noting that multiplication in~$M$ is constant-time since~$L$ is fixed.

  Now, if this update completes the block, we insert it into the macro
  structure. To do this, we first create the block as a singleton
  block interval in the macro layer, whose monoid value is the value of the
  newly computed block. Then we compute the maximal interval of complete blocks
  to which the block belongs, again by attempting to merge its singleton
  interval with the 
  neighboring interval on the right and the neighboring interval on the left
  if they exist. These macro merges take
  constant time because the interval lengths and monoid values are stored at
  interval endpoints.

  At the end of the stream, all blocks form a single macro interval. Its
  stored monoid value is $\mu(w)$, from which we can deduce whether $w\in L$.

  The overall space bound is $O(n)$ bits: the macro arrays have
  $O(n/\log n)$ entries of one RAM cell each, while the packed micro bit vectors
  and monoid vectors use $O(\log n)$ bits per block.
  We use lazy array evaluation~\cite{grandjean2023which} to ensure that our data
  structures are initially zero without the need for running a preprocessing
  phase.
\end{proofsketch}

\begin{proof}
	Let $L$ be the target regular language.
        Let $M$ be its syntactic monoid and
        $\mu\colon \Sigma^*\to M$ be its syntactic morphism (see
        \cref{sec:monoid} in the main text for the formal definitions).
	We assume throughout the proof that $|M|$ is much smaller than $n$;
        indeed, for all small values of $n$ we can simply use a brute force algorithm
	without changing the asymptotic complexity.

	We partition the input word $w$
	of length~$n$ into blocks of length $b\coloneq \lceil\log n\rceil$, except that the last
	block may be shorter. This yields $B=\lceil n/b\rceil$
	blocks.
	The value of $b$ is computed in constant time using the
	most-significant-bit operation available in our RAM model, and $B$ is then
	computed in constant time by integer division with remainder, which is
        also available in our RAM model. Each position $1 \leq i \leq n$ of the
        input word corresponds to a block and to an offset
	within this block. Namely, position $i$ corresponds to the block
	$\lfloor (i-1)/b \rfloor+1$ and to the offset $(i-1) \bmod b + 1$.
	The algorithm has two layers.

        \medskip
	\textbf{Macro layer.}
	Blocks are treated as atomic symbols. A block is called \emph{complete} once all
	positions that correspond to this block have been streamed.
        The macro layer uses two arrays
	of length $B$:
	\begin{itemize}
		\item a table $T_{\text{size}}$ of size $B$, indexed by block indexes
			$1,\ldots,B$, whose entries are RAM cells storing interval lengths;
	  \item a table $T_\mu$ of size $B$, also indexed by block indexes
		  $1,\ldots,B$, whose entries are RAM cells storing elements of~$M$ (which fit since $|M|$ is fixed and smaller than $n$).
	\end{itemize}
	These tables maintain the monoid value and the interval size for each 
	maximal contiguous interval of complete blocks. More
        precisely, the invariant
	is the following: for every maximal contiguous interval $[i,j]$ of complete
	blocks, where $1 \leq i \leq j \leq B$ are block indexes, all blocks $\beta_k$ with
	$i\leq k \leq j$ are complete, while $\beta_{i-1}$ and $\beta_{j+1}$ are
	either not complete or out of bounds. We store in $T_{\text{size}}$ at
	positions $i$ and $j$ the size $j-i+1$ of the interval
	and in $T_\mu$ at position $i$ 
	the monoid value of the factor covered by this interval.
        More precisely, if
	$x_k$ is the image under~$\mu$ of the word
	formed by the completed block $\beta_k$, i.e. $x_k = \mu(w[(k-1)b+1]
        \cdots w[(k-1)b+b_k])$,
        then $T_\mu[i]$ stores the product $x_i\cdots x_j$ in~$M$.
	Note that the invariant only imposes requirements on the
        values of $T_\mu$ and $T_{\text{size}}$ stored at endpoints of maximal contiguous
          intervals of complete blocks; there is no requirement imposed on
          the other cells as they will never be queried. In particular, we do not
          need to erase values of $T_\mu$ and $T_{\text{size}}$ when positions stop being endpoints of
          a maximal contiguous interval of complete blocks: as intervals only grow,
          we know that such positions will never be queried again as endpoints.

        \medskip

	\textbf{Micro layer.}
	Inside each block, we maintain information analogous to the macro layer,
	but instead of maintaining the interval sizes we simply maintain a bitvector
	representing revealed letters.
	Let $1 \leq j \leq B$ be a block index.
        We write $b_j$ for the
	length of block~$\beta_j$,
        so $1 \leq b_j \leq b$.
        More precisely, for each block
	$\beta_j$ we store:
	\begin{itemize}
	  \item a bit vector $v_j$ of length $b_j$ indicating which positions of $\beta_j$
		  have already been streamed; in this bit vector, the bit at
                  position $1$ is the leftmost one, which corresponds to the least significant bit;
	  \item a vector $m_j$ whose entries are indexed by the offsets
	  $1,\ldots,b_j$ of $\beta_j$. For every maximal contiguous streamed
	  interval $[p,q]$ inside $\beta_j$, the entry $m_j[p]$ stores the image
	  under~$\mu$ of the factor revealed at offsets $p,\ldots,q$.
	\end{itemize}
	The entries of $m_j$ are elements of the fixed monoid $M$, so they can be stored
	using a constant number of bits. Thus, the whole vector $m_j$ is represented
	using a constant number of RAM cells, and its entries can be accessed and
	updated in constant time.


	The bit vectors $v_j$ for $1 \leq j \leq B$ of the micro
        layer need to be initially 0, to denote that no position
        of the word has been streamed initially.
        To ensure this, we use
        the \emph{lazy array evaluation technique} described in
	\cite[Section 2.5]{grandjean2023which}, which lets us use the bit
        vectors as if they all had been initialized to zero in constant time
        (in particular avoiding the need for a preprocessing phase).

        \medskip
	\textbf{Handling streamed letters (overall).}
	When a triple $(a,r,n)$ with $1 \leq r \leq n$ is streamed, 
	we first determine the block $\beta_j$ containing position~$r$ and the offset
	$p$ of this position within block~$\beta_j$, 
	i.e., $j \coloneq \lfloor (r-1)/b \rfloor + 1$ and $p \coloneq (r-1) \bmod b + 1$,
        by performing integer division with remainder by~$b$ in constant
        time.
	We also compute $b_j$, which is $b$ if $j < B$ and $n \bmod b$ otherwise.
        We then first handle this update
        in~$\beta_j$ at the micro layer, as we now explain, before handling an
        update at the macro layer if this completes the last position
        of~$\beta_j$.

        \medskip
	\textbf{Handling streamed letters (micro layer).}
	What we do here intuitively
        amounts to what we would do to achieve a space complexity of $O(b_j)$
        memory cells on a block of length $b_j$, except that to achieve a
        space complexity of $O(b_j)$ bits we have
        packed all
        information about the block $\beta_j$ in the constant number of cells
        used by $\beta_j$. The main challenge is to argue that we can indeed
        handle the streamed letter in constant time using the operations allowed
        by the RAM model, in particular paying attention to the fact that we
        must find the endpoints of intervals by direct examination of the bit
        vector $v_j$ (i.e., we do not have an analogue of the table
        $T_{\text{size}}$ of the macro layer, intuitively because we do not have
        the space to store it).

        More precisely, to handle a streamed letter in $\beta_j$, we first initialize a singleton interval at position~$p$
        of~$\beta_j$,
        by setting the $p$-th bit of~$v_j$ to be~$1$ and setting the $p$-th
        element of $m_j$ to be $\mu(a)$ for $a$ the received character. 
	This can
        be performed in constant time in our RAM model: indeed, there is only a
        constant number of possible images of $\mu$, and the values being
        manipulated can be constructed in constant time via bit shifts.

        We then attempt to merge this singleton interval with an interval starting
        at the right-neighboring position (if it exists). More precisely, if $p = b_j$ then we
        do nothing because $p$ is the rightmost position of the block.
        Otherwise, we check the $(p+1)$-th bit of~$v_j$. If it is~$0$, then we do nothing
        because there is no interval covering that position. Otherwise, we merge
        the newly created singleton interval with the interval that exists at
        its right. To do this, we set the $p$-th
        value of~$m_j$ to be the product in~$M$ of $\mu(a)$ and of the element at the
        $(p+1)$-th position of~$m_j$. 
	All of this can still be performed in constant time in
        our RAM model. In particular, the product in~$M$ can be done in constant
        time because $M$ is of constant size.

        Finally, we try to merge the resulting interval with an interval ending at
        the left-neighboring position (if it exists).
        To do this, if $p = 1$ we do
        nothing because $p$ is the leftmost position of the block. Otherwise, we
        check the $(p-1)$-th bit of~$v_j$, and do nothing if it is~$0$ because
        there is no interval that ends at position $p-1$. If the $(p-1)$-th bit
        of~$v_j$ is~$1$,
        we need to merge with the left interval.
        This is more challenging than the merge with a right-neighboring interval described
        in the previous paragraph, because now 
        we must efficiently retrieve the image of the left interval by $\mu$
        from $m_j$, and for this we need to locate the
        position of the left endpoint $p'$ of the interval. Once we have
        determined $p'$, we will set the $p'$-th
        position of $m_j$ by multiplying in~$M$ the current value (namely, the value of
        the left interval) with the value of the $p$-th position of~$m_j$ (namely,
        the value of the current interval, which is either the singleton
        interval at $p$ or the result of merging that interval with its right
        neighboring interval).
        It only remains to explain how to find $p'$ in constant time.

        We first spell out the definition of $p'$, before explaining how this
        definition can be implemented in constant time in our model.
        If there is no position $1 \leq q < p$ whose bit in $v_j$ is~$0$, then
        the left interval starts at $p' \coloneq 1$ because there is no gap
        before position~$p$. Otherwise, let $1 \leq q < p-1$ 
      be the greatest (i.e., rightmost)
        position whose bit in $v_j$ is~$0$; the left interval then starts at $p'
        \coloneq q+1$.
        To implement this computation efficiently, we shift $v_j$ to the right
        by $b_j-p$ positions, so that only the $p$ bits strictly to the left of
        position~$p$ remain. We then bitwise negate this shifted word, 
	and apply the most-significant-bit operation. 
        If the resulting word is zero, then
        there is no such position~$q$ and $p'=1$; otherwise, the
        most-significant-bit operation returns the index $t$ of the rightmost
        1 (i.e., the rightmost 0 before bitwise negation), then undoing the shift gives
	the corresponding original offset~$q = t-(b_j-p)$,
        and we set $p'=q+1$.
        All these operations are constant-time
        operations in our RAM model.

        \medskip
	\textbf{Handling streamed letters (macro layer).}
        We have explained how the micro structure is updated in constant time,
        and we must now explain how to do the same with the macro structure. We
        only need to update the macro structure when the updated block $\beta_j$
        becomes complete, i.e., when all positions belonging to $\beta_j$ have
        their bit set to~$1$ in $v_j$. This is checked in constant time 
	by taking the bitwise xor of $v_j$ with $2^{b_j}-1$ 
        and checking whether the result is $0$.
	Note that to compute $2^{b_j}-1$ without overflow we can compute
        it as $2^{b_j - 1}+(2^{b_j - 1}-1)$, which can be computed in constant
        time in our RAM model using constants, shifts, predecessor, and
        addition.
        Then, the micro layer ensures that the value $x$ of $m_j$ at position
        $0$ is the value in~$M$ of the contiguous factor of~$w$ covered by the
        block~$\beta_j$.

        We now insert $x$ as a singleton interval of the macro layer by setting
        $T_\mu[j] \coloneq x$ and $T_{\text{size}}[j] \coloneq 1$. We must then
        attempt to merge this singleton block with the macro intervals that may exist
        immediately to the right and immediately to the left. This is similar to
        the micro structure but it is simpler because we use (and update) the table
        $T_{\text{size}}$ which avoids the need for bit-vector endpoint search;
        we give the details below.

        During these two merge attempts, we keep in a temporary variable the monoid
        value of the interval currently containing block~$j$. Initially, this
        value is $x$ and the current interval has length~$\ell=1$. 

        We first try to merge with an interval immediately to the right, if one
        exists. If $j=B$, or if $\beta_{j+1}$ is not complete (checked as
        previously described by taking xor with $2^{b_{j+1}}-1$ and testing that
	the result is nonzero, using the same technique to compute this constant without overflow), 
        then there is no such interval and the current monoid value and length are unchanged.
        Otherwise, this right interval starts at $j+1$, so its monoid value is
        stored in $T_\mu[j+1]$ and its length is stored in
        $T_{\text{size}}[j+1]$. We multiply the current monoid value, which is
        initially $x$,
        with the right value $T_\mu[j+1]$, and the result becomes the current monoid value of the
        interval starting at~$j$. We also add the length of the right interval
        to the current length~$\ell$. Last, we update $T_{\text{size}}$ at the two
        endpoints of the merged interval. These endpoints are the index $j$ of
        the newly inserted
        block and the previous right endpoint of the right interval.

        We then try to merge with the interval immediately to the left. If
        $j=0$, or if $\beta_{j-1}$ is not complete (checked as previously
        described by taking xor with $2^{b_{j-1}}-1$ and testing that the result
	is nonzero, using the same technique to compute it without overflow),
        then there is no
        such interval, so we store the current monoid value in $T_\mu[j]$.
        Otherwise, the length stored in $T_{\text{size}}[j-1]$ allows us to
        compute where the left interval starts: letting
        $s \coloneq T_{\text{size}}[j-1]$, then this left endpoint is $i
        \coloneq j-s$. We multiply
        the monoid value stored in $T_\mu[i]$ with the current monoid value, and
        store the result in $T_\mu[i]$. We also update $T_{\text{size}}$ at the
        two endpoints of the merged interval to store the length $s+\ell$ of the
        merged interval. These endpoints are $i$ and the right endpoint
        $j+\ell-1$ of the current interval.

        \medskip
        \textbf{Space usage.}
        The macro layer stores two arrays of length $B$. Each entry fits in one
        RAM cell, so this uses $O(B\log n)=O(n)$ bits.
        The micro bit vectors use
        $Bb=O(n)$ bits in total. The vectors $m_j$ also use $O(Bb)=O(n)$ bits,
        because each entry is an element of the fixed monoid~$M$ and therefore
        has constant size. Hence the total space usage is $O(n)$ bits.

        \medskip
        \textbf{End of the stream.}
        When the stream is finished, all blocks are complete, and the macro
        level contains a single interval, so that $T_\mu[0]$ must contain the
        monoid element achieved by that interval, i.e., by the entire word~$w$.
        From this, we can determine whether $w \in L$, which concludes the proof.
\end{proof}

In light of \cref{word_upper_n}, in the remainder of the paper we focus on space
complexity. We only study deterministic space complexity.
We always account for space complexity in terms of the number of bits used (not
memory cells), and we only consider the space that must be memorized between any
two elements of the out-of-order stream; we do not consider the computation time and
interim memory needed when processing an element.

\subparagraph*{Fooling sets for lower bounds.}
We now introduce \emph{fooling sets} as a general technique to
prove unconditional lower bounds, and showcase fooling sets to show
a linear space lower bound on the specific language $\Sigma^* aa \Sigma^*$.
Our fooling sets are an out-of-order analogue of standard techniques 
(see, e.g., \cite[Definition~2]{Hospodar2018} for a similar notion)
used for
the state complexity of regular languages when reading words from left to right,
e.g., to show that any deterministic automaton for the language $L_n = \Sigma^*
a \Sigma^n$ must have at least $2^{n+1}$ states
(see
\cite[Example~2.13]{hopcroft2006introduction} which proceeds along similar
lines).

Let $\ept$ be a fresh placeholder symbol, and 
write $\bar{X} \coloneq X \cup \{\ept\}$ for a set extended by~$\ept$.
A \textit{partial word} over an alphabet $\Sigma$ 
is a word $w$ over the alphabet $\bar{\Sigma}$.
The \emph{domain} of~$w$ is the subset $\dom(w)$ of the positions where $w$ does not
carry a placeholder, i.e., $\dom(w) \coloneq \{i \in \{1, \ldots, |w|\} \mid
w[i] \in \Sigma\}$.
Two partial words $w_1,w_2$ are \emph{homogeneous} if $|w_1| = |w_2|$ and $\dom(w_1) =
\dom(w_2)$, and a set $S$ of partial words is \emph{homogeneous} if all
its partial words are pairwise homogeneous.
Now, two partial words $w_1,w_2$ are 
\textit{compatible} if they have the same length and if $\dom(w_1) \cap
\dom(w_2) = \emptyset$, i.e.,
for all $1 \leq i \leq |w_1|$ at least one of $w_1$ and $w_2$ has a placeholder
at position~$i$.
Their \textit{composition} $w_1 \circ w_2$ is then the partial word defined by $(w_1
\circ w_2)[i] \coloneq w_1[i]$ if $i \in \dom(w_1)$, 
and $(w_1
\circ w_2)[i] \coloneq w_2[i]$ otherwise. We further call two compatible words $w_1$ and $w_2$
\textit{complementary} if $\dom(w_1) \cup \dom(w_2) = \{1, \ldots, n\}$,
equivalently, $w_1 \circ w_2$
is a word of~$\Sigma^*$ without placeholders.
Given a language $L \subseteq \Sigma^*$, a word length $n$,
and two homogeneous partial words $w_1,w_2 \in \bar{\Sigma}^n$,
we call \emph{complementary witness} of $w_1,w_2$ for~$L$ a partial word $v \in \bar{\Sigma}^*$
complementary to both $w_1$ and $w_2$,
such that precisely one of $w_1 \circ v$ and $w_2 \circ v$ belongs to~$L$.

For $L$ a language over~$\Sigma$ and $n,m \in \mathbb{N}$,
an \textit{$(m,n)$-fooling set for~$L$}
is a homogeneous set $W = \{w_0,\dots,w_{m-1}\}$ of $m$ partial words $w_i \in \bar{\Sigma}^n$ 
such that any two 
distinct partial words $w_i,w_j \in W$
have a complementary witness for~$L$.

%

\begin{lemmarep}[Fooling sets lower bound]
\label{thm:fooling_lower_bound}
Let $L$ be a fixed language over alphabet $\Sigma$, and let
$f\colon \mathbb{N}\to\mathbb{N}$ be a non-decreasing function. 
If for every integer $N \in \mathbb{N}$ there exists $n > N$ such that there is
an $(f(n),n)$-fooling set for~$L$, then the out-of-order membership problem
to~$L$ has $\Omega(\log f(n))$ space complexity.
\end{lemmarep}

The proof of \cref{thm:fooling_lower_bound} is standard and deferred to the
appendix
: we simply argue that the
algorithm must have enough memory to distinguish each partial word among those
of the
fooling set.
Note that, in the result statement above and everywhere in the paper, we write
$g = \Omega(f)$ for two functions $f, g \colon \mathbb{N}_{>0} \to
\mathbb{N}_{>0}$ to mean that $\limsup_{n \rightarrow \infty} g(n) / f(n) >
0$. This is the 
$\Omega$ notation of Hardy and Littlewood~\cite{HardyLittlewood1914}
and not the one by Knuth~\cite{Knuth1976BigOmicron} (which is defined as meaning
$f = O(g)$).
For instance, letting $g$ be the function defined by
$g(n) \coloneq n$ for odd $n$ and $g(n) \coloneq 1$ otherwise, we have $g =
\Omega(n)$. This is important, e.g., to show lower bounds on out-of-order
membership for languages that are hard for words of odd length and trivial for
words of even length.


\begin{proof}
  Fix $L$ as in the statement. It suffices to show that for each
  $N\in \mathbb{N}$ there is $n > N$ such that, on words of length~$n$, any
  out-of-order membership algorithm requires $\log f(n)$ memory bits. The
  lemma statement guarantees us the existence of an $n > N$ for which there is
  an $(f(n),n)$-fooling set
	$W_n = \{w_0,\dots,w_{f(n)-1}\}$ for $L$. Order arbitrarily the common
        domain of the partial words of~$W_n$, e.g., order it increasingly as $p_1
        < \cdots < p_\ell$; and further let $q_1 < \cdots < q_r$ be the elements
        of the complement of the domain; note that $\ell + r = n$.

	Let $A$ be any (deterministic) algorithm solving the out-of-order
        membership problem for $L$ on words of length~$n$, and let us show that
        $A$ needs $\log f(n)$ memory bits.
        For each word $w_i \in W_n$, let $s_i$ be the state of the memory after
        the partial out-of-order stream $P_i \coloneq (w_i[p_1], p_1, n), \ldots,
        (w_i[p_\ell], p_\ell, n)$.
        It suffices to show that the $s_i$ are pairwise distinct, because
        this implies that we need $\geq \log |W_n|$ bits to write them, i.e., we
        use at least $\log f(n)$ memory bits as required.
        
        To see why, let us proceed by contradiction and assume that we have
        two integers $i \neq j$ for which $s_i = s_j$. 
        Let $v_{ij}$ be the complementary witness of~$w_i$ and~$w_j$, and
        consider the continuation of the stream
        $Q = (v_{ij}[q_1], q_1, n), \ldots,
        (v_{ij}[q_r], q_r, n)$. We immediately see that the concatenated
        streams $P_i Q$ and $P_j Q$ are respectively out-of-order streams for
        $w_i \circ v_{ij}$ and $w_j \circ v_{ij}$, both of which follow the
        permutation whose images are first $p_1, \ldots, p_\ell$ and then
        $q_1, \ldots, q_r$. Now, the definition of complementary
        witnesses ensures that exactly one of $w_i \circ v_{ij}$ and $w_j \circ
        v_{ij}$ belongs to~$L$, so $A$ must reply differently on $P_i Q$
        and on $P_j Q$. But we assumed that $A$ was in the same state $s_i =
        s_j$ after $P_i$ and $P_j$, so it is still in the same state after
        reading the continuation $Q$. This is a contradiction and
        concludes the proof.
\end{proof}

Let us showcase \cref{thm:fooling_lower_bound} and prove that some languages
require linear space for out-of-order membership:

\begin{proposition}\label{mon_notricom_lin}
  The out-of-order membership problem for the language $\Sigma^* aa \Sigma^*$
  on alphabet $\Sigma = \{a,b\}$
  requires linear space.
\end{proposition}

\begin{proof}
        For each length $n$, we define a $(2^n,3n)$-fooling set $W_n$ for $\Sigma^* aa
        \Sigma^*$  formed of the $2^n$ partial words defined by the
        regular expression $(b (a+b) \ept)^n$. Note that they are all
        homogeneous and have length $3n$. Let $w_i, w_j$ be two distinct partial words
        of~$W_n$. There must be a position $0 \leq p < n$ such that $w_i[3p+2]
        \neq w_j[3p+2]$. Note that this means one is $a$ and the other is~$b$;
        without loss of generality up to exchanging $w_i$ and $w_j$ we assume
        that $w_i[3p+2] = a$ and $w_j[3p+2] = b$. Now, we let $v_{ij} \coloneq
        (\ept \ept b)^{p-1} (\ept \ept a) (\ept \ept b)^{n-p}$. We see that
        $v_{ij}$ is complementary to both $w_i$ and $w_j$. Further, by
        construction $w_i \circ v_{ij}$ belongs to $\Sigma^* aa
        \Sigma^*$ because there are two contiguous $a$'s at positions $3p+2$
        and $3p+3$. By contrast, $w_j \circ v_{ij}$ does not belong to
        $\Sigma^* aa \Sigma^*$. This establishes that $W_n$ is indeed a
        $(2^n,3n)$-fooling set for $\Sigma^* aa \Sigma^*$ and concludes the
        proof.
\end{proof}

We will use fooling sets
to show all the lower
bounds in this paper; in particular these bounds are
all unconditional. Note that this
also implies that all our lower bounds already hold
for streaming orders that are chosen in advance, i.e., when we show a lower
bound on a language for words of length~$n$ then we show the existence of an
$n$-permutation
such that the space lower bound applies even on algorithms that work
only for that specific permutation.  We also note that lower bounds shown with
this method can be understood as a lower bound on fixed-partition communication
complexity for a one-way protocol, where Alice has access to the letters on the
domain of the fooling set and Bob has access to the rest of the word.

\subparagraph*{Constant space upper bound for commutative languages.}
We complement the previous result by noticing that the linear space complexity
lower bound does not apply to all languages: we can sometimes
solve out-of-order membership with better space complexity. 
We will show the easy fact that, for \emph{commutative} regular languages,
out-of-order membership can be solved with constant space. Here, a language
is \emph{commutative} if, for every word $w$, letting $n \coloneq |w|$, for
every $n$-permutation $\pi$, we have $w \in L$ if and only if $w[\pi(1)] \cdots
w[\pi(n)] \in L$. Then:


\begin{lemma}\label{com_const}
	For any fixed commutative regular language $L$,
there is an algorithm for
out-of-order membership to~$L$ with 
\(O(1)\) space complexity.
\end{lemma}

\begin{proof}
  Fix a deterministic automaton $A$ that accepts $L$. Now, whenever we receive
  some tuple $(a,i,n)$, we feed $a$ to~$A$. This takes constant time, and the memory of the algorithm stores
  the automaton state, which has constant size because $L$ is fixed. At the end of the stream, we accept
  if $A$ accepts. It is immediate that $A$
  was fed a permutation of the word~$w$: as $L$ is commutative, the
  algorithm is correct, which concludes.
%
%
\end{proof}

The condition of being commutative is sufficient but is not necessary to have
constant-space complexity, as illustrated below:
\begin{example}
\label{ex:abstar}
   There is an algorithm with $O(1)$ space complexity for out-of-order
   membership to the language $(ab)^*$. Indeed, we simply need to check that 
   all streamed $a$'s are at an odd position, 
	that all streamed $b$'s are at an even position, and that the word
        length is even.
\end{example}

In summary,
we have shown that out-of-order membership to fixed regular languages can
always be solved in constant time per streamed character and linear memory,
and this memory bound is optimal for some languages (e.g., $\Sigma^* aa
\Sigma^*$) but suboptimal for others (e.g., commutative languages or $(ab)^*$
for which the
bound is $O(1)$). Thus, our goal is to classify regular languages according
to the space complexity in this sense. 
Towards this, we will present algebraic variants of
the problem on which classifications are easier to show.

\section{Out-of-Order Monoid Evaluation}
\label{sec:monoid}
In this section we focus on \emph{monoids}, i.e., finite sets equipped with an
associative composition law and a neutral element. We consider the
\emph{out-of-order evaluation problem} for monoids, which can be seen as a
special case of out-of-order membership for regular languages.
In this section, we first define this problem, and then
show that tractable monoids for this problem
are closed under the monoid variety operators. We then show a classification
(\cref{thm:monoid}) of the
possible space complexity regimes.

\subparagraph*{Out-of-order evaluation.}
All monoids in this paper are implicitly finite.
Let $M$ be a fixed monoid, let $n>0$ be the word length,
and let $w = w[1] \cdots w[n]$ be a word over $M$ of length $n$; its
\emph{product} is the value $w[1] \cdot \cdots \cdot w[n]$ of~$M$ obtained
according to the composition law of~$M$.
We define \emph{out-of-order streams}
like in the previous section. The
\emph{out-of-order evaluation problem} for~$M$ is the following task:
letting $w$ be a word of length~$n$, letting $\pi_n$ be an~$n$-permutation, and
letting $R$ be the out-of-order stream on~$w$ defined by~$\pi_n$,
we receive the elements of~$R$ one after the other and must compute at the
end the product of~$w$ in~$M$.

We note that the out-of-order evaluation problem for a monoid $M$ can be
expressed in terms of out-of-order membership problems. Indeed, for any $x \in
M$, the language $L_x$ over $M$ defined as the words of $M^*$ having product~$x$
is always a regular language. 
Likewise, the out-of-order membership problem for a language $L$ reduces to the
out-of-order evaluation problem for the \emph{syntactic monoid} of~$L$, i.e.,
the finite monoid obtained by quotienting $\Sigma^*$ by the \emph{syntactic
congruence} equivalence relation $\sim_L$ defined by~$L$. Intuitively, $\sim_L$
equates words whose behavior relative to~$L$ is the same, and formally we have $u \sim_L v$ for two
words $u,v \in \Sigma^*$ whenever for any $s,t\in\Sigma^*$ we have $sut \in L$
if and only if $svt \in L$. It is well-known that for a regular language $L$ the
syntactic congruence $\sim_L$ has indeed a finite number of equivalence classes,
and in fact this property precisely characterizes the regular languages;
see, e.g., \cite{straubing2021varieties} for more details.

Thus, algorithms for out-of-order evaluation for the syntactic monoid $M$ of a
language~$L$ imply an algorithm for out-of-order membership to~$L$. However,
there are cases where out-of-order membership for~$L$ is easier than the same
problem for~$M$, e.g., because $L$ may not feature a neutral
letter but $M$ always has a neutral element. For this reason, 
classifying the complexity of
out-of-order evaluation for monoids is a coarser task
than classifying the complexity of out-of-order membership for regular
languages, and so it serves as a useful first step.

%

To start our study of out-of-order evaluation for monoids, let us observe that the complexity classes of
monoids for out-of-order evaluation are closed under the monoid variety operators.

\subparagraph*{Closure properties.}
The \emph{direct product} of two monoids $M_1$ and $M_2$ is the monoid $(M_1
\times M_2, \cdot)$ with the binary operator defined as
$(x_1, x_2) \cdot (x_1', x_2') = (x_1 \cdot x_1', x_2 \cdot x_2')$ and identity
$(e_1, e_2)$ for $e_1$ and $e_2$ the respective identities of~$M_1$ and~$M_2$.
A \emph{sub-monoid} of a monoid $M$ with neutral element~$e$
is a subset $M'\subseteq M$ of~$M$ which is closed
under the composition law of~$M$ and which is also a monoid with neutral
element~$e$.
A \emph{congruence} over a monoid~$M$ is an equivalence relation $\sim$ over $M$
which satisfies the following requirement: whenever $s \sim t$ and $s' \sim t'$
for $s, t, s', t' \in M$ then we have $ss' \sim tt'$.
A \emph{quotient monoid} $M/\sim$ is the quotient of a monoid $M$ by a
congruence $\sim$ over~$M$, i.e., it consists of the equivalence classes of~$M$
under~$\sim$, its composition law is lifted to these classes, and its neutral
element is the equivalence class of the neutral element of~$M$. Note that,
thanks to the definition of a congruence, when applying the composition law the
results do not depend on
the choice of representatives.

%
A class of monoids is a \emph{pseudo-variety}, only called \emph{variety}
throughout the paper,
if it is closed under direct products, taking sub-monoids, and taking monoid quotients.
Let $\mathbf V$ and $\mathbf W$ be two varieties of monoids.
Their \emph{join}, denoted $\mathbf V \lor \mathbf W$, is the smallest variety
containing both the monoids of~$\mathbf V$ and those of~$\mathbf W$.

We show that, for the problems that we study, each complexity class forms a
variety:

\begin{propositionrep}
\label{closure_mono}
For any function $f\colon \mathbb{N}\to\mathbb{N}$,
	the set of monoids admitting an $O(f(n))$ space algorithm for
        out-of-order evaluation forms a monoid variety, i.e., 
        it is closed under direct product, sub-monoid, and monoid quotient.
\end{propositionrep}

\begin{proofsketch}
  Any data structure for the out-of-order evaluation problem for a
  monoid~$M$ can in particular solve it for any submonoid, and it can also solve
  it for any quotient (simply by picking arbitrary representatives and returning the equivalence class of the result).
  Further, 
  for any two monoids $M_1$ and $M_2$, we can solve out-of-order evaluation for
  $M_1 \times M_2$ by simultaneously using one data structure for $M_1$ and one
  for~$M_2$.
\end{proofsketch}

\begin{proof}
	Fix $f$ as in the lemma statement. We show each closure property in
        turn.

        Let $M$ be a monoid that admits an $O(f(n))$ space algorithm $A$ for
        out-of-order monoid evaluation.
	We first observe that running the exact same algorithm $A$ solves the 
	out-of-order monoid evaluation problem for any sub-monoid $M'$ of~$M$, i.e.,
        if we operate under the assumption that all letters of the input are
        in~$M'$, then the result of~$A$ will be an element of~$M'$ which is the
        correct result.

        Now, take a congruence $\cong$ over~$M$, and consider $M / \cong$. If
        the letters of the input are in $M/\cong$, i.e., they are equivalence
        classes of~$\cong$, then we lift them to letters of~$M$ by picking any
        representative in~$M$ of the equivalence class. The result of~$A$  will
        then be an element of~$M$ whose equivalence class is the correct result
        in~$M / \cong$. This gives an algorithm for out-of-order evaluation of $M /
        \cong$ with same space usage as~$A$.

        Last, let $M_1$ and $M_2$ be two monoids that admit an $O(f(n))$ space
        complexity algorithms for out-of-order monoid evaluation, written $A_1$
        and $A_2$ respectively.
	Let $R$ be an out-of-order stream of length $n \in \mathbf{N}$ for the
        out-of-order monoid evaluation problem of $M_1 \times M_2$.
	When an element $(s_1 \times s_2, i, n)$ is received, we stream $(s_1,i,n)$ to $A_1$ and $(s_2,i,n)$ to $A_2$.
	At the end of the stream, letting $r_1$ and $r_2$ be the results from
        $A_1$ and $A_2$, then the desired answer is simply $(r_1, r_2)$.
	The total space required is $O(f(n) + f(n)) = O(f(n))$.
\end{proof}

\subparagraph*{Space complexity of out-of-order monoid evaluation.}
Our goal in this section is to classify all monoids $M$ according to the space needed by
an algorithm to solve out-of-order evaluation for~$M$. Our results will imply
the existence of precisely three complexity regimes (constant, logarithmic, and linear),
and characterize the corresponding monoids.

Let us first introduce the varieties that constitute the different regimes. The
first one is $\mathbf{Com}$, the variety of \emph{commutative monoids}. We
can characterize it as those monoids $M$ satisfying the equation $xy = yx$ for
all $x, y \in M$, i.e., we write $\mathbf{Com} = \llbracket xy = yx \rrbracket$.
In fact, by Reiterman's theorem~\cite{reiterman1982birkhoff}, we know that varieties of monoids
always admit such equational descriptions (see~\cite{straubing2021varieties}
for more details), so we will define 
our other monoid varieties using equations in this fashion.
Monoids from the variety $\mathbf{Com}$ admit a constant-space algorithm
for the out-of-order evaluation problem, obtained by
immediately lifting \cref{com_const} to monoids.
We will show that this in fact characterizes all constant-space monoids, with all other monoids having a logarithmic lower bound.

For the logarithmic space regime, we will characterize it as the join of
$\mathbf{Com}$ together with another variety, called $\mathbf{FL}$ (for
``first-last''). To give an equation for $\mathbf{FL}$, we need to define the
\emph{idempotent power} of a monoid~$M$ as
the smallest non-negative integer $\omega$ 
such that
$x^\omega = x^{2\omega}$ for all $x \in M$.
Then we can define $\mathbf{FL} = \llbracket (xy)^\omega st (xz)^\omega = (xy)^\omega sxt
(xz)^\omega \rrbracket$. One example of a monoid in $\mathbf{FL}$ is 
the syntactic monoid of the language
$\Sigma^*a\Sigma^*b\Sigma^*$
over the alphabet $\{a,b,e\}$, intuitively because membership can be decided by
remembering the position of the leftmost $a$ and of the rightmost~$b$.
 We will show that monoids in
$\mathbf{FL}$ admit a logarithmic-space algorithm for out-of-order evaluation,
which will imply the same for $\mathbf{FL} \lor \mathbf{Com}$ thanks to
\cref{closure_mono}.
It turns out that $\mathbf{FL} \lor \mathbf{Com}$ 
precisely coincides with the class $\mathbf{W}$
in \cite{tesson} that forms the tractability boundary of a similar but seemingly
incomparable problem, namely, the monoids with logarithmic simultaneous deterministic communication
complexity.

We will last characterize the linear space regime by showing,
using an equational characterization of $\mathbf{FL} \lor \mathbf{Com}$,
that monoids outside of this class require
linear space.

All told, in this section, we show the following characterization, summarized in
the first column of \cref{all_table} (with all upper bounds matched by tight lower
bounds) and examples (given as syntactic monoids of languages) in the first row
of \cref{all_ex_table}:

\begin{theorem}
  \label{thm:monoid}
  Let $M$ be any fixed monoid. Then the space complexity of out-of-order monoid
  evaluation for~$M$ is as follows:
  \begin{itemize}
    \item If $M \in \mathbf{Com}$, then it is in~$O(1)$
    \item If $M \in (\mathbf{FL} \lor \mathbf{Com}) \setminus \mathbf{Com}$,
      then it is in $\Theta(\log n)$
    \item Otherwise, it is in $\Theta(n)$.
  \end{itemize}
\end{theorem}

We prove this result in the rest of the section.
Recall that the first point follows by the immediate analogue of
\cref{com_const} for monoids, for instance:

\begin{example}
   \label{ex:parity}
   The syntactic monoid of the language \((aa)^*\),
   which is the group \(\mathbb{Z}/2\mathbb{Z}\), is commutative and
   therefore admits a constant-space upper bound.
\end{example}

We give first the logarithmic lower bound for
non-commutative monoids (for the lower bound of the second point), then the
logarithmic upper bound (for the upper bound of the second point), and then the linear
lower bound (for the lower bound of the third point). The upper bound of the
third point is immediate by the monoid analogue of \cref{word_upper_n}.

\subparagraph*{Logarithmic lower bound.}
For our two lower bounds in this section, we immediately lift the definitions of
\emph{fooling sets} to the out-of-order monoid evaluation problem in the
expected way: the words are defined over the alphabet $\bar{M}$ featuring
the monoid elements and the placeholder, and the only new stipulation is that
any two words $w_i \neq w_j$ of the fooling set have a complementary witness
$v_{ij}$ that ensures that $w_i \circ v_{ij}$ and $w_j \circ v_{ij}$
evaluate to different monoid elements (instead of saying that precisely one
belongs to the target language).

We then show the first lower bound, establishing that non-commutative monoids need at least
logarithmic space complexity:

\begin{lemmarep}\label{mon_log_lower}
For any fixed non-commutative monoid $M \notin \mathbf{Com}$,
its out-of-order evaluation problem has $\Omega(\log n)$ space complexity.
\end{lemmarep}

\begin{proofsketch}
        Let $x, y \in M$ be two elements such that $x y \neq yx$, and let $e$ be
        the neutral element.
	We use \cref{thm:fooling_lower_bound} 
        (lifted to the setting of monoids)
        and define for each $n > 0$ the
        fooling set $W_n$ formed of the words
        $w_i \coloneq (e\ept)^{i-1}\, (x\ept)\, (e\ept)^{n-i}$
        for each $1 \leq i \leq n$. The complementary
        witnesses for $1 \leq i < j \leq n$ are defined as: 
        $v_{ij} \coloneq (\ept e)^{i-1} \, (\ept y) \, (\ept e)^{n-i}$,
        ensuring that $w_i \circ
        v_{ij}$ evaluates to $xy$ and $w_j \circ v_{ij}$ to~$yx$.
\end{proofsketch}

\begin{proof}
	Since $M$ is non-commutative, pick $x,y \in M$ such that $xy \neq yx$,
        and denote by~$e$ the neutral element of~$M$.
	For any $n>0$, let us define a $(n,2n)$-fooling set $W_n$. 
	Define the partial word $w_i \in \bar{M}^{2n}$ for each $i \in
        \{1,\dots,n\}$ by:
        \[w_i = (e \ept)^{i-1} \, (x\ept) \, (e\ept)^{n-i}\]
        For each $1 \leq j \leq n$, define the partial word:
        \[v_{j}' = (\ept e)^{j-1} \, (\ept y) \, (\ept e)^{n-j}\]
        It is clear by definition that, for any $1 \leq i \leq j \leq n$, we
        have $w_i \circ v_j' = xy$, and that for any $1 \leq j < i \leq n$ we
        have $w_i \circ v_j' = yx$. Thus, for any $1 \leq i < j \leq n$
        we can pick as \textit{complementary witness} for $w_i$ and $w_j$ the
        word $v_{ij} \coloneq v_i'$, so that $w_i \circ v_{ij} = xy \neq yx = w_j
        \circ v_{ij}$. 
        Thus,
        by the immediate analogue of \cref{thm:fooling_lower_bound} for
        monoids, the out-of-order monoid evaluation problem for $M$ 
	has $\Omega(\log n)$ space complexity.
\end{proof}

\subparagraph*{Logarithmic upper bound.}
We now show the logarithmic upper bound for languages in $\mathbf{FL} \lor \mathbf{Com}$.
Thanks to the closure under variety operations (\cref{closure_mono}), using
again the constant upper bound for commutative monoids (\cref{com_const}),
it suffices to show an upper bound for $\mathbf{FL}$.

Remember our equational definition of $\mathbf{FL}$ using the equation $\mathbf{FL} =
\llbracket (xy)^\omega st (xz)^\omega = (xy)^\omega sxt (xz)^\omega \rrbracket$.
We will now show that, whenever a monoid $M$ satisfies this equation,
then computing the product of elements of~$M$ can be achieved simply by looking at
a constant number of leftmost and rightmost occurrences of each element. 
Let us introduce the requisite definitions on words over~$M$:

\begin{definition}[First-last subword]
  \label{def:fl}
  Let $M$ be a monoid and let $m \in M$.
	Let $u$ be a word over~$M$ of size $n$.
        The \emph{$m$-occurrences of~$u$} is
	the subset $O_m(u) \subseteq [n]$ of the positions $i \in [n]$ such that
        $u[i] = m$. 
	We then define the \emph{first $k$ positions} $\FF^k_m(u)$ as the subset of $O_m(u)$ 
        formed of the $k$ smallest positions:
        note that if there are no more than $k$ occurrences of~$m$ in~$u$ 
        then $\FF^k_m(u) = O_m(u)$.
        We define the \emph{last positions} $\LL^k_m(u)$
        analogously with the $k$ greatest positions.
	Now, the \emph{$k$-first-last subword} of~$u$ is the scattered subword 
        $\FL^k(u)$ formed by only retaining the first and last $k$ occurrences for each $m$, i.e.
	retaining the monoid elements at the positions $\bigcup_{m \in M}
        \FF^k_m(u) \cup \LL^k_m(u)$.
\end{definition}

We note that the first-last subword defined above also occurs in~\cite{tesson},
where it is called $\text{RED}_k$. Their work then extends this definition to
correspond to their variety $\mathbf{W}$, which will correspond to $\mathbf{FL}
\lor \mathbf{Com}$ as we explain later. By contrast, we will use first-last
subwords directly for the simpler variety $\mathbf{FL}$ for our upper bound, and
the extension to $\mathbf{FL} \lor \mathbf{Com}$ is for free thanks to 
\cref{closure_mono} and \cref{com_const}.
We claim the following property on monoids in~$\mathbf{FL}$:

\begin{claimrep}
  \label{clm:flfirstlast}
  Let $M$ be a monoid in $\mathbf{FL}$, and let $k \coloneq |M|$.
  Then for any word $u \in M^*$, letting
  $v \coloneq \FL^k(u)$, the words $u$ and $v$ have the same product in~$M$.
\end{claimrep}

\begin{toappendix}
  The proof of \cref{clm:flfirstlast} relies on the following auxiliary lemma
  from the main text, which we now restate and prove:
\end{toappendix}

To show this, let us point out how to use the $\mathbf{FL}$ equation on words of
elements of~$M$. Using standard pumping arguments, we can show:

\begin{claimrep}
  \label{clm:usepump}
  Let $M$ be a monoid in $\mathbf{FL}$, let $k \coloneq |M|$,
  let $m \in M$, and let $u = s v t$ be a word of~$M^*$ such that $s$
  and $t$ contain exactly $k$ occurrences of~$m$. Let $v'$ be the scattered
  subword of~$v$ obtained by removing all occurrences of~$m$, and let $u'
  \coloneq s v' t$. Then $u$ and $u'$ have the same product in~$M$.
\end{claimrep}

\begin{proof}
Using the assumption, decompose $s$ as $s = s_0 m s_1 m \cdots s_{k-1} m s_k$,
with $s_0, \ldots, s_k$ being words over~$M$. Consider the prefixes $s_i'
\coloneq s_0 m \cdots s_i$ for each $0 \leq i \leq k$. By the pigeonhole
principle, there are two $0 \leq i < j \leq k$ such that $s_i'$ and $s_j'$ have
the same product in~$M$. 

Let us decompose $s = s_- m s' s_+$, where
we set $s_- \coloneq s_0 m \cdots s_{i-1} m s_i$,
we set $s' \coloneq s_{i+1} \cdots m s_j$, and
we set $s_+ \coloneq m s_{j+1} \cdots m s_k$. 

Let us write $\eta$ the canonical morphism from~$M^*$ to~$M$ that achieves
monoid evaluation, i.e., for $w$ a word over~$M$ the value $\eta(w)$ denotes the product of~$w$
in~$M$. We have by definition (*): $\eta(s) = \eta(s_-) \eta(m s') \eta(s_+)$ by definition, and we
claim the following identity (**): $\eta(s) = \eta(s_-) (\eta(m) \eta(s'))^\omega \eta(s_+)$.
Indeed, we know from the choice of~$i$ and~$j$
that $\eta(s_- m s') = \eta(s_-)$, i.e., $\eta(s_-) \eta(ms') = \eta(s_-)$.
Injecting the equation into itself $\omega$ times,
we get $\eta(s_-) (\eta(ms'))^\omega = \eta(s_-)$, altogether we have $\eta(s_-)
\eta(m s') = \eta(s_-) (\eta(ms'))^\omega$, which we can inject in (*) to obtain
(**).

By applying the same reasoning on $t$, we can decompose $t = t_- m t' t_+$ such
that $\eta(t) = \eta(t_-) (\eta(m) \eta(t'))^\omega \eta(t_+)$.

Now we are ready to use the $\mathbf{FL}$ equation. We have:
\[
  \eta(u) = \eta(s) \eta(v) \eta(t) = \eta(s_-) (\eta(m) \eta(s'))^\omega
  \eta(s_+) \eta(v) \eta(t_-) (\eta(m) \eta(t'))^\omega \eta(t_+)
\]
We now show that any occurrence of~$m$ in~$v$ can be eliminated without changing
the product, ie., the image by~$\eta$. Write $v = v_- m v_+$. We have:
\[
  \eta(u) = \eta(s_-) \big[(\eta(m) \eta(s'))^\omega
    (\eta(s_+) \eta(v_-)) \eta(m) (\eta(v_+) \eta(t_-)) (\eta(m) \eta(t'))^\omega
\big] \eta(t_+)
\]
Now, 
recall that $\mathbf{FL} = \llbracket (xy)^\omega st (xz)^\omega = (xy)^\omega sxt
(xz)^\omega \rrbracket$.
Applying this equation from right to left to the right-hand side of the
expression above, we obtain:
\[
  \eta(u) = \eta(s_-) \big[(\eta(m) \eta(s'))^\omega
    (\eta(s_+) \eta(v_-)) (\eta(v_+) \eta(t_-)) (\eta(m) \eta(t'))^\omega
\big] \eta(t_+)
\]
We can repeat this process to eliminate all occurrences of~$m$ in $v$ one by one
until none remain. At the end of the process, we obtain:
\[
  \eta(u) = \eta(s_-) \big[(\eta(m) \eta(s'))^\omega
    \eta(s_+) \eta(v') \eta(t_-) (\eta(m) \eta(t'))^\omega
\big] \eta(t_+)
\]
Substituting back, this implies $\eta(u) = \eta(s) \eta(v') \eta(t)$, our
desired claim.
\end{proof}

This claim allows us to show \cref{clm:flfirstlast}:

\begin{proof}[Proof sketch of \cref{clm:flfirstlast}]
  We do an induction on the number of elements of~$M$ that occur more than $2k$
  times in~$u$. If there are none, then $v = u$ are equal and the claim is
  immediate. Otherwise, we pick an element $m$ occurring more than $2k$ times,
  we isolate its first and last $k$ occurrences, we use \cref{clm:usepump} to
  argue that all other occurrences can be eliminated without changing the
  product in~$M$, and we conclude using the induction hypothesis on the
  resulting word which has one less element occurring more than~$2k$ times.
\end{proof}

\begin{toappendix}
We are now ready to show \cref{clm:flfirstlast}:

\begin{proof}[Proof of \cref{clm:flfirstlast}]
	Let $M \in \mathbf{FL}$ be a monoid and let $k$ be its size.
        Let us show by induction on~$\ell$ the following claim: for any word 
        $u$ over $M$ where there are $\ell$ letters of~$M$ appearing more than $2k$
        times, letting $v \coloneq \FL^k(u)$, the words $u$ and $v$ have the
        same product in~$M$.

        The base case of $\ell = 0$ is immediate because we then have $u = v$.

        For the induction, fix $\ell > 0$ and assume that the claim holds
        for~$\ell-1$. 
	Pick some element $m \in M$ which appears more than $2k$ times in $u$.
	Split $u$ in $3$ subwords $u = u_0 \cdot u_1 \cdot u_2$,
        such that $u_0$ and $u_2$ respectively contain the $k$ first and $k$
        last occurrences of~$m$ in $u$.
        Let $u_1'$ be the scattered subword of~$u_1$ obtained by removing all
        the occurrences of~$m$, and let $u' \coloneq u_0 \cdot u_1' \cdot u_2$.
        We see that $\FL^k(u) = \FL^k(u')$, intuitively because none of the
        positions removed in~$u'$ relative to~$u$ were kept in $\FL^k(u)$.
        Hence, we can deduce the following claim (*): $v = \FL^k(u)$ and $v' \coloneq \FL^k(u')$ have the same product
        in~$M$ (because, as we argued, they are the same word).
        Further, $u'$ has $\ell-1$ letters appearing more than $2k$ times, so we
        can apply the induction hypothesis to it: letting $v' \coloneq
        \FL^k(u')$, we have the following claim (**): the product of $u'$ and $v'$ in~$M$ is
        identical.
        Last, 
        using \cref{clm:usepump} to $u_0 \cdot u_1 \cdot u_2$ where $u_0$ and
        $u_2$ each contain $k$ occurrences of~$m$, we know that $u$ and $u'$
        have the same product in~$M$.
        Combining (*) and (**) and (***), we know that $u$ and $v$ have the same
        product in~$M$. This completes the induction step and concludes the
        proof.
\end{proof}
\end{toappendix}

We can now state our upper bound on out-of-order evaluation  for monoids
in~$\mathbf{FL}$:

%
%
%
%
%
%
%
%

\begin{lemmarep}\label{mon_log_up}
	For any fixed monoid $M\in\mathbf{FL}$, 
	the out-of-order evaluation problem for $M$ has $O(\log n)$ space complexity.
\end{lemmarep}

\begin{proofsketch}
  The algorithm remembers the $k$ minimal and the $k$ maximal positions of each
  element of~$M$ during the stream. This allows it to construct $\FL^k(u)$ for
  $u$ the streamed word, and thanks to \cref{clm:flfirstlast} the product of
  $\FL^k(u)$ in~$M$ is identical to the desired answer, namely, the product
  of~$u$ in~$M$.
\end{proofsketch}

\begin{proof}
	Let $M \in \mathbf{FL}$ be a monoid, and let
        $k \coloneq |M|$ be its size.
        Recall that $M$ is fixed, so $k$ is a
        constant.
        The algorithm works as follows on a word $u$ of length~$n$: it maintains 
        for each element $m\in M$ of the monoid,
        we maintain the positions of its first (leftmost) $k$ occurrences, and
        of its last (rightmost) $k$ occurrences.
        Note in particular that if
        $m$ occurs $2k$ times or less, then we store the positions of all its
        occurrences.
        More precisely, whenever we receive a tuple $(m,i,n)$, then we compare
        $i$ to the
        $2k$ stored occurrences of the letter~$m$, and we insert it at the
        proper place in the two sorted lists.
        This takes logarithmic space because, for each element in~$M$ (whose
        size is fixed), we store a constant number of occurrences (namely,
        $2k$), each of which
        takes $O(\log n)$ bits to write down.
        (Incidentally, note that the algorithm
        also takes constant time per streamed symbol.)
        At the end of the stream, we argue that the information stored allows us
        to determine the $k$-first-last subword of the word $u$. Indeed, the
        algorithm knows the sets $\FF^k_m(u)$ and $\LL^k_m(u)$ for each $m \in
        M$, thus by definition it can compute $\FL^k(u)$ (and this computation
        takes constant time and space). Now, we know by
        \cref{clm:flfirstlast} that $u$ and $\FL^k(u)$ have the same product
        in~$M$, so the algorithm simply returns the product in~$M$ of
        $\FL^k(u)$, again computed in constant time and space because $\FL^k(u)$
        has constant length. This concludes the proof.
\end{proof}

\begin{example}
	\label{ex:ab}
	The syntactic monoid of the language \(ab\), which
	is the five-element monoid \(\{1, a, b, ab, 0\}\), belongs to
        \(\mathbf{FL}\) because it satisfies the $\mathbf{FL}$ equation:
	\[
		(xy)^\omega st (xz)^\omega = (xy)^\omega sxt (xz)^\omega.
	\]
	Indeed, it has two idempotents, \(0\) and \(1\). If one of the
	\(\omega\)-powers evaluates to \(0\), the equation holds trivially
	since both sides evaluate to \(0\). The only remaining case is
	\((xy)^\omega = (xz)^\omega = 1\), which implies \(x=1\), and then both
	sides evaluate to \(st\).

\end{example}

\subparagraph*{Linear lower bound.}
The only remaining piece to show \cref{thm:monoid} is the linear lower bound on
monoids not in $\mathbf{FL} \lor \mathbf{Com}$. Formally:

\begin{lemmarep}\label{mon_lin}
For any fixed monoid $M \notin \mathbf{FL} \lor \mathbf{Com}$,
the out-of-order evaluation problem for $M$ has $\Omega(n)$ space complexity.
\end{lemmarep}

To show this result, we need an equation characterizing the variety $\mathbf{FL} \lor \mathbf{Com}$, which we will then use to obtain witnessing elements to
build fooling sets. 
It turns out that  $\mathbf{FL} \lor \mathbf{Com}$ coincides with the class
$\mathbf{W}$ in \cite{tesson}, so we will reuse part of their proof; however
our approach has the merit of proposing a simpler equation for the class, and
of giving a presentation of $\mathbf{W}$
as a join involving a simpler variety $\mathbf{FL}$ with an equational
characterization. We also remind the reader that the results of~\cite{tesson}
pertain to communication complexity; as far as we understand, it is a
coincidence that tractability boundaries of our problem coincides with one of
theirs, because there does not appear to be reductions from one problem to the
other.

We accordingly give the following characterization of $\mathbf{FL} \lor
\mathbf{Com}$:

\begin{lemma}\label{fleq}
	Let $M$ be a monoid. The following conditions are equivalent:
	\begin{enumerate}
		\item The monoid $M$ belongs to $\mathbf{FL} \lor \mathbf{Com}$.
		\item The monoid $M$ satisfies the equation $(xa)^\omega sxtu (xb)^\omega = (xa)^\omega stxu (xb)^\omega$.
                \item The monoid $M$ satisfies the equations  of the
                  monoid variety $\mathbf{W}$
                  in~\cite{tesson}, namely:
			\begin{itemize}
				\item $(swt)^\omega xw^\omega y (uwv)^\omega = (swt)^\omega xy (uwv)^\omega$
				\item $(swt)^\omega xw (uwv)^\omega = (swt)^\omega wx (uwv)^\omega$
			\end{itemize}
	\end{enumerate}
\end{lemma}

\begin{proofsketch}
	We show $1\to2\to3\to1$.
	The $1\to2$ direction follows from the equations of $\mathbf{FL}$ and $\mathbf{Com}$.
	The $2\to3$ direction is easy by equation manipulation.
	For $3\to1$ we use the $RED_{t,p}$ congruence defined in \cite{tesson}.
        We reuse their proof \cite[Theorem~3]{tesson} which shows that any monoid
        satisfying the equations of~$\mathbf{W}$ is refined by a
        $RED_{t,p}$-congruence: this is the analogue of the result we showed
        for~$\mathbf{FL}$ in \cref{clm:flfirstlast}.
        We can then conclude because $RED_{t,p}$ is the intersection of an
        $\mathbf{FL}$ congruence and of a $\mathbf{Com}$ congruence.
\end{proofsketch}

\begin{toappendix}
  To show our characterization of $\mathbf{FL} \lor \mathbf{Com}$ (\cref{fleq}),
  we first define the \emph{$k$-first-last relation}:

  \begin{definition}
	The \emph{$k$-first-last relation} $\sim_{FL,k}$ on $\Sigma^*$ is defined by
	letting $u \sim_{FL,k} v$ whenever $\FL^k(u) = \FL^k(v)$. The \emph{$k$-first-last
	class} $[u]_{\FL^k}$ of a word $u \in \Sigma^*$ is the language of the words
	$v$ such that $v \sim_{FL,k} u$. An \emph{FL language} is a language $L \subseteq
	\Sigma^*$ that can be written as a finite union of $k$-first-last classes,
	i.e., there are words $u_1, \ldots, u_\ell \in \Sigma^*$ such that $L =
	\bigcup{1 \leq i \leq \ell} [u_i]_{\FL^k}$.
\end{definition}

  We first verify that the $k$-first-last relation is a congruence:

\begin{claimrep}
  \label{clm:flkcong}
  The $k$-first-last relation is a congruence, i.e., it is an equivalence relation and further $u \sim_{FL,k} v$ and
  $u' \sim_{FL,k} v'$ implies $uu' \sim_{FL,k} vv'$. 
\end{claimrep}

\begin{proofsketch}
	Note that $\sim_{FL,k}$ is an equivalence relation, which follows directly from its definition.
	Now observe that the $i$-th occurrence of any letter in $uu'$ appears in 
	$\FL^k(uu')$ if and only if it appears in $\FL^k(\FL^k(u)\FL^k(u'))$ and thus
	$ \FL^k(uu') = \FL^k(\FL^k(u)\FL^k(u')) = \FL^k(\FL^k(v)\FL^k(v')) = \FL^k(vv') $.
	Hence $uu' \sim_{FL,k} vv'$, and thus $\sim_{FL,k}$ is a congruence.
\end{proofsketch}

\begin{proof}
	We first note that $\sim_{FL,k}$ is an equivalence relation, which follows directly from its definition.
	We now prove that $\sim_{FL,k}$ is a congruence on $\Sigma^*$.
	Let $u,v,u',v' \in \Sigma^*$ be words over $\Sigma$ with $u \sim_{FL,k} v$ and $u' \sim_{FL,k} v'$,
	we want to show that $uu' \sim_{FL,k} vv'$.
	By definition, this means $\FL^k(u) = \FL^k(v)$ and $\FL^k(u') = \FL^k(v')$

	We claim that $\FL^k(uu')=\FL^k(\FL^k(u)\FL^k(u'))$.
	Let $a_i$ be the $i$-th occurrence of $a$ in $uu'$ with $i \leq k$.
	By definition $a_i$ will be the $i$-th $a$ $in \FL^k(uu')$
	Let $j$ be the number of occurrences of $a$ in the word $u$.
	If $i$ is smaller than $j$ then $a_i$ will be the $i$-th $a$ in the word $u$,
	thus also be the $i$-th $a$ in $\FL^k(u)$, 
	thus also be the $i$-th $a$ in $\FL^k(\FL^k(u)\FL^k(u'))$.
	If $i$ is bigger than $j$, then $a_i$ will be the $i-j$-th $a$ in the word $u'$,
	thus be the $i-j$-th $a$ in $\FL^k(u')$, 
	thus be the $i$-th $a$ in $\FL^k(\FL^k(u)\FL^k(u'))$.
	Using the same arguments we show the equality for the $k$-last occurrences.
	Therefore $\FL^k(uu')=\FL^k(\FL^k(u)\FL^k(u'))$.

	Using this identity, we obtain
	$ \FL^k(uu') = \FL^k(\FL^k(u)\FL^k(u')) = \FL^k(\FL^k(v)\FL^k(v')) = \FL^k(vv') $.
	Hence $uu' \sim_{FL,k} vv'$, and thus $\sim_{FL,k}$ is a congruence.
\end{proof}

We also need the easy observation that monoids refined by the $\FL^k$-congruence
must be in the class $\mathbf{FL}$:

\begin{lemma}
  \label{lem:fltofl}
	Let $M$ be a monoid, and $k \coloneq \omega$
        be the idempotent power of $M$.
        Assume that, 
	for all words $v,w$ over $M$ such that $\FL^k(v) = \FL^k(w)$, we have
        that the product of $v$ in~$M$ is equal to the product of~$w$ in~$M$.
        Then
	$M$ is a monoid belonging
        to the variety $\mathbf{FL}$.
\end{lemma}

\begin{proof}
  To show that $M$ is in $\mathbf{FL}$, letting $x,a,b,s,t,u$ be elements in
  $M$, we must show that
  $(xa)^\omega sxtu (xb)^\omega = (xa)^\omega stxu (xb)^\omega$.
  Thus, defining the words over~$M$
  $w \coloneq (xa)^\omega sxtu (xb)^\omega$ and 
  $v \coloneq (xa)^\omega stxu (xb)^\omega$, we must show that their product
  in~$M$ is identical. For this, from our assumption on~$M$, it suffices to show
  that $\FL^k(v) = \FL^k(w)$.
  By definition we have:
  \[\FL^k(w) =
    \FL^k((xa)^\omega sxtu (xb)^\omega)
    \FL^k((xa)^\omega stu (xb)^\omega)
  \]
  because $x$ appears at least $k$ times before and after the occurrence we have
  removed. For the same reason, $\FL^k(v) = \FL^k((xa)^\omega stu (xb)^\omega)$.
  Hence, $\FL^k(v) = \FL^k(w)$, which suffices to conclude.
%
%
\end{proof}

We now have all the ingredients to prove \cref{fleq}:

\begin{proof}[Proof of \cref{fleq}]
	We will do a $1\to2\to3\to1$ style proof.
	The $1\to2$ direction follows from the equations of $\mathbf{FL}$ and $\mathbf{Com}$.
	The $2\to3$ direction follows from the equation $(xa)^\omega sxtu (xb)^\omega = (xa)^\omega stxu (xb)^\omega$.
	For $3\to1$ we use the $RED_{t,p}$ congruence defined in \cite{tesson}
	to show that any monoid can be decomposed into a commutative congruence and in a
	$\mathbf{FL}$ type congruence based on the $k$-first-last occurrences of each letter.

	We first show the $1\to2$ direction. We first show that 
	all $\mathbf{FL}$ monoids satisfy the equation of~$2$. Indeed, using the
        equation of $\mathbf{FL}$,
	we have $(xy)^\omega sxtu (xz)^\omega = (xy)^\omega stu (xz)^\omega = (xy)^\omega stxu (xz)^\omega$.
	We then observe that $\mathbf{Com}$ monoids obviously satisfy the equation.
        Hence, considering the variety $\mathbf{V}$ of the monoids satisfying the equation
        of~$2$, we know that $\mathbf{V}$ contains $\mathbf{FL}$ and it contains $\mathbf{Com}$. As
        $\mathbf{FL} \lor \mathbf{Com}$ is the smallest variety containing
        $\mathbf{FL} \lor \mathbf{Com}$, it means that $\mathbf{FL} \lor
        \mathbf{Com}$ is included in $\mathbf{V}$. Hence, we have concluded that
        all monoids of $\mathbf{FL} \lor \mathbf{Com}$ satisfy the equation
        of~2.

        \medskip

	We now show $2\to 3$ direction.
	Assume $(xa)^\omega sxtu (xb)^\omega = (xa)^\omega stxu (xb)^\omega$.
	Let us first show that this implies another equation that we will use
        afterwards, namely (*):
        \[(a'xa)^\omega sxtu (b'xb)^\omega = (a'xa)^\omega stxu (b'xb)^\omega\]
        For this, we have:
        \begin{align*}
          (a'xa)^\omega sxtu (b'xb)^\omega 
              & = 
          (a'xa)^\omega (a'xa)^\omega sxtu (b'xb)^\omega (b'xb)^\omega  &
          \text{by idempotence}\\
              & = 
          a' (xaa')^\omega (xaa')^{\omega-1} xa sxtu b' (xbb')^\omega
          (xbb')^{\omega-1} xb  &
          \text{by conjugation}\\
              & = 
          a' (xaa')^\omega (xaa')^{\omega-1} xa stxu b' (xbb')^\omega
          (xbb')^{\omega-1} xb  &
          \text{by the eqn.\ of~2}\\
              & = 
          (a'xa)^\omega (a'xa)^\omega stxu (b'xb)^\omega (b'xb)^\omega  &
          \text{by conjugation}\\
              & = 
          (a'xa)^\omega stxu (b'xb)^\omega 
          \text{by idempotence}
        \end{align*}

        We can then show the first equation using (*):
	\begin{align*}
		(swt)^\omega xw^\omega y (uwv)^\omega 
			&= (swt)^\omega (swt)^\omega xw^\omega y (uwv)^\omega &\text{by idempotence}\\
			&= (swt)^\omega (st)^\omega xw^\omega w^\omega y
          (uwv)^\omega\quad &\text{moving $w$ one by one with (*)}\\
			&= (swt)^\omega (st)^\omega x w^\omega y (uwv)^\omega\quad &\text{by idempotence}\\
			&= (swt)^\omega (swt)^\omega x y (uwv)^\omega\quad &\text{moving them back one by one}\\
			&= (swt)^\omega x y (uwv)^\omega\quad &\text{by idempotence}\\
	 \end{align*}
	The second is immediate as it is simply equation (*) in the specific
        case where $s$ and $u$ are the neutral element. Hence, we have shown the
        equation of~3.

	We now show the $3\to 1$ direction.
	Let $M$ be a monoid in $\llbracket (swt)^\omega xw^\omega y (uwv)^\omega 
			= (swt)^\omega xy (uwv)^\omega ,
			(swt)^\omega xw (uwv)^\omega 
			= (swt)^\omega wx (uwv)^\omega \rrbracket$.
        Following~\cite{tesson}, we define the equivalence relation $RED_{t,p}$
        on words over~$M$: two words $\alpha, \beta$ are equivalent under
        $RED_{t,p}$ if $\FL^t(\alpha) = \FL^t(\beta)$ (this is written $RED_t$
        in~\cite{tesson}) and if for every letter $m \in M$ the number of
        occurrences of~$m$ in~$\alpha$ and in~$\beta$ has the same remainder
        modulo~$p$. Intuitively, this equivalence relation is a refinement of
        ``having the same image by~$\FL^t$'' to also account for the commutative
        information tracked by $\mathbf{Com}$ monoids. By 
        \cite[Theorem~3]{tesson}, as $M$ satisfies the equation of their
        variety $\mathbf{W}$,
        there is a $RED_{t,p}$ congruence which
        captures $M$, i.e., the equivalence relation on words of~$M^*$ that
        equates words with the same product in~$M$ is refined by a $RED_{t,p}$
        congruence for some $t$ and $p$. (Intuitively, this is similar to the
        claim of \cref{clm:flfirstlast}, which applied to $\mathbf{FL}$
        monoids and to the equivalence relation defined by~$\FL^k$.) We now show
        that this implies that $M$ is in $\mathbf{FL} \lor \mathbf{Com}$. 
        For this,
        notice that the classes of the $RED_{t,p}$ congruence are by definition
        intersections of equivalence classes of the $k$-first-last relation
        (which is a congruence, cf \cref{clm:flkcong}) on the one hand, and of
        the \emph{modulo-$p$ congruence} that equates two words having the same number of
        occurrences of every element of~$M$ modulo~$p$ on the other hand (it is
        immediate that this is also a congruence). Now, the equivalence classes of the congruence
        defined by~$M$ on~$M^*$ (which equates two words having the same product
        in~$M$) are unions of such intersections. This implies that $M$ is a
        quotient of $M_1 \times M_2$, where $M_1$ and $M_2$ are the monoids
        defined respectively by the equivalence classes of the $k$-first-last
        congruence and of the modulo-$p$ congruence. By \cref{lem:fltofl}, $M_1$
        is a monoid in~$\mathbf{FL}$, and $M_2$ is obviously a monoid
        of~$\mathbf{Com}$. Hence, $M$ is a quotient of a product of one monoid
        in~$\mathbf{FL}$ and one monoid in~$\mathbf{Com}$, which concludes.
\end{proof}
\end{toappendix}

With this characterization, we can now prove the linear lower
bound (\cref{mon_lin}):

\begin{proof}[Proof of \cref{mon_lin}]
        Let $M$ be a monoid not in $\mathbf{FL} \lor \mathbf{Com}$. By
        \cref{fleq}, there are elements in $M$ witnessing a violation of the
        equation: $(xa)^\omega sxtu (xb)^\omega \neq (xa)^\omega stxu
        (xb)^\omega$.
	We use \cref{thm:fooling_lower_bound} and build a $(2^n,5n+2)$-fooling set
        for any $n>0$.
	For each $\alpha \in\{0,1\}^n$ we define $w_{\alpha_i}$ by either
        $\ept x \ept e \ept$ or $\ept e \ept x \ept$ 
	depending on the bit $\alpha_i$, let
	$w_\alpha \coloneq \ept w_{\alpha,1} w_{\alpha,2} \dots w_{\alpha,n}
        \ept$, and let $W_n \coloneq \{w_\alpha \mid \alpha \in \{0,1\}^n\}$ be
        the fooling set.

        Now, for any $\alpha \neq \beta$, let $\iota(\alpha,\beta)$ be a
        position of $\{1, \ldots, n\}$
        such that $\alpha_{\iota(\alpha,\beta)} \neq \beta_{\iota(\alpha,\beta)}$. The complementary witness is
        then:
	\[v_{\alpha\beta} = ((xa)^{n\omega - \iota(\alpha,\beta)+1}) ~ (e \ept e
          \ept a)^{\iota(\alpha,\beta)-1} ~ (s \ept t \ept u) ~
        (e \ept e \ept b)^{n-\iota(\alpha,\beta)} ~
      ((xb)^{n\omega-n+\iota(\alpha,\beta)})\]
        where the first and last groups in parentheses correspond to a single
        element.

        We can then see that one of $w_\alpha \circ v_{\alpha\beta}$ 
        and $w_\beta \circ v_{\alpha\beta}$
        is 
        $(xa)^{n\omega} (sxtu) (xb)^{n\omega} = (xa)^\omega sxtu (xb)^\omega$
        and the other  is
        $(xa)^{n\omega} (stxu) (xb)^{n\omega} = (xa)^\omega stxu (xb)^\omega$,
        so by assumption they are indeed different, showing that $W_n$ is a
        fooling set.
	This implies that the out-of-order monoid evaluation problem for $M$ 
	has $\Omega(n)$ space complexity.
\end{proof}

We illustrate the lower bound of~\cref{mon_lin} with the following example:

\begin{example}
	\label{ex:abba}
	We show that the syntactic monoid of the language \(M(a^*bba^*)\) is not in the
	class \(\mathbf{FL}\lor\mathbf{Com}\), and therefore admits an
	\(\Omega(n)\) space lower bound. This monoid has seven elements \(\{1,
	a, b, ab, ba, bb, 0\}\), where \(1\), \(a\), and \(0\) are idempotent,
	and satisfies \(bbb=0=aba\).

	It does not satisfy the equation of~\cref{fleq}\[ (xy)^\omega sxtu (xz)^\omega = (xy)^\omega
	stxu (xz)^\omega, \] where we rename the variables of the equation to
	avoid confusion with the monoid elements. Taking \(x=y=z=a\), \(s=1\),
	and \(t=u=b\), the left-hand side evaluates to \(abba = bb\), while the
	right-hand side evaluates to \(ababa = 0\).

	Hence, by~\cref{mon_lin}, the out-of-order monoid evaluation problem for
	\(M(a^*bba^*)\) requires \(\Theta(n)\) space.
\end{example}

\section{Constant-Space Out-of-Order Semigroup Evaluation}
\label{sec:semigroup}
In the previous section, we characterized the out-of-order evaluation problem
for monoids, i.e., the case of out-of-order membership for languages with a neutral
letter. This assumption on languages is a common one, and it is made throughout
in~\cite{tesson}. We now study what happens when this assumption is removed. For
this, we turn to the out-of-order evaluation problem for finite
\emph{semigroups} instead of monoids. In this section, we first define the problem and comment on
the relationship to monoids and languages. We then show our main result for
semigroups, which characterizes the constant-space class.

%

\subparagraph*{Out-of-order semigroup evaluation.}
All semigroups in this paper are implicitly finite.
The definition of out-of-order evaluation for semigroups is defined exactly
like for monoids, but with a semigroup. By the immediate analogue of \cref{closure_mono}, the
complexity regimes of out-of-order semigroup evaluation form varieties of
semigroups. Further, the definitions of fooling sets and the corresponding lower
bound (\cref{thm:fooling_lower_bound}) also immediately adapt. The
\emph{syntactic semigroup} of a language $L$ is defined like the syntactic monoid
but as a quotient of $\Sigma^+$, i.e., without including the empty word.
(However, it may have a neutral element nevertheless, for instance if $L$
features a neutral letter.)

Like monoid evaluation, the out-of-order semigroup evaluation problem is a coarsening of out-of-order
membership, but it is less coarse, in particular it does not require the presence of a neutral
letter.
We note that our results on monoids already imply lower bounds 
on some semigroups. For instance, semigroups that are not \emph{locally
commutative}, i.e., which admit a subsemigroup which is a non-commutative
monoid, directly inherit the logarithmic space lower bound on
monoids (\cref{mon_log_lower}) from the previous section:

\begin{lemma}
  \label{lem:lcom}
  Let $S$ be a fixed semigroup. If $S$ is not \emph{locally commutative}, i.e.,
  if
  $S$ does not satisfy the equation
   $s^\omega x s^\omega y s^\omega = s^\omega y s^\omega x s^\omega$, 
   then the out-of-order evaluation problem for~$S$ has $\Omega(\log n)$ space complexity.
\end{lemma}
\begin{proof}
	Let $S$ be a semigroup and $s,x,y\in S$ such that 
	$s^\omega x s^\omega y s^\omega \neq s^\omega y s^\omega x s^\omega$.
	Let $x'=s^\omega x s^\omega$ and $y'=s^\omega y s^\omega$. 
	We remark that the sub-semigroup generated by $\{s^\omega, x', y'\}$
	is a monoid with $s^\omega$ as its neutral element. Further, it is non-commutative
	because $x'y' \neq y'x'$. Hence, its evaluation problem has $\Omega(\log
        n)$ space complexity,
	and thus the evaluation problem of $S$ also has $\Omega(\log n)$ space
        complexity
        by closure under sub-semigroups (\cref{closure_mono}).
\end{proof}

\begin{example}
\label{ex:sg_lower_bound_aba_aca}

The semigroup $S(a^* b a^+ c a^*)$ has elements $\{a, b, c, ba, ac, bac, 0\}$,
	with $a$ and $0$ being the only idempotents.  It does not satisfy the
	equation $s^\omega x s^\omega y s^\omega = s^\omega y s^\omega x
	s^\omega$ when taking $s = a$, $x = b$, and $y = c$: the right side
	evaluates to $0$ while the left side evaluates to $bac$.  The elements
	$a$, $ba$, and $ac$ generate a local submonoid $\{a, ba, ac, bac, 0\}$,
	which is not commutative.  Hence, by~\cref{lem:lcom}, 
        the out-of-order evaluation problem
                        for this semigroup has $\Theta(\log n)$ space
                        complexity.

			Note that we can also devise an $O(\log n)$ space
                        algorithm for this semigroup
by observing that in any word, at most one
	element containing $b$ (i.e., $b$, $ba$, or $bac$) and at most one
	element containing $c$ (i.e., $c$, $ac$, or $bac$) can appear.  If there
        are more, or if there is a $0$, then the result is~$0$. If there are one
        or less of each of these positions, then all
	remaining positions are necessarily filled with $a$'s.
        The algorithm can therefore store the
	exact positions of the elements containing $b$ and of the elements
        containing $c$ (and any $0$ element if present), rejecting if more than
        two elements are memorized, and otherwise 
	computing the product of the memorized elements according to the order
        of their positions, filling
	gaps with $a$'s as needed.  Since there are only a constant number of
	positions to memorize and each position can be represented with $O(\log
	n)$ bits, this gives an $O(\log n)$ space upper bound.
                        Thus, the out-of-order evaluation problem
                        for $S(a^* b a^+ c a^*)$ has $\Theta(\log n)$ space
                        complexity.
\end{example}

However, the space complexity of out-of-order evaluation for semigroups does not
exactly follow the classification on monoids, and new phenomena can arise. For
instance:

\begin{propositionrep}
	\label{ex:asigmab}
  Let $L$ be the language $a\Sigma^* b$.
   The out-of-order evaluation problem for the syntactic semigroup
   $S$ of~$L$
   has $O(1)$ space complexity.
\end{propositionrep}

\begin{proofsketch}
  We just memorize the element at position 1 and the element at position~$n$ for
  $n$ the word length, using constant space; and we return their product.
\end{proofsketch}

\begin{proof}
   Let us give an explicit description of $S$.
   It consists of four elements, which we will write as $(x,y)$ for $x,y\in
   \{a,b\}$, with the composition law 
   $(x, y)\cdot (x', y') = (x, y')$. 
   The letter $a$ corresponds to the element $(a,a)$ and the letter $b$ to the
   element $(b,b)$, and we accept the words mapped to~$(a,b)$.

   The algorithm for out-of-order evaluation for this semigroup is defined as
   follows, on words of length~$n$.
   We simply memorize the letter $x$ at position $1$ and the letter $y$ at
   position $n$. At the end, we return the semigroup element $(x,y)$.
   This clearly uses
   constant memory and is correct because we know that the entire word would be
   mapped to the same semigroup element: this is thanks to the fact that the
   semigroup satisfies the equation $xzy=xy$.
\end{proof}

This example shows that, for out-of-order semigroup evaluation, it may be useful
to remember the first and last positions of the overall word of semigroup
elements. 
This is not helpful for monoids, intuitively because the first and last
positions can always contain the neutral element.
The class of semigroups whose evaluation only depends on the leftmost and
rightmost elements is a well-studied variety called
$\mathbf{Li}$ and defined by the equation $\llbracket x^\omega y x^\omega =
x^\omega \rrbracket$, or equivalently
$\mathbf{Li} = \llbracket x^\omega y z^\omega = x^\omega z^\omega\rrbracket$.

The main result which we show in this section is the following:

\begin{theorem} \label{low_log_sem}
	Let $S$ be any fixed semigroup. Then the space complexity of
        out-of-order semigroup
  evaluation for~$S$ is as follows:
  \begin{itemize}
    \item If $S \in \mathbf{Li} \lor \mathbf{Com}$, then it is $O(1)$
    \item Otherwise, it is in $\Omega(\log(n))$
  \end{itemize}
\end{theorem}

We prove this result in the rest of the section. We first show the first point
(the upper bound), before showing the second point (the lower bound). We note
that the second point only gives a lower bound, which is not tight relative to
our linear space upper bound (\cref{word_upper_n}); we come back to this issue in
\cref{sec:horror}.

\subparagraph*{Upper bound.}
We first prove the upper bound for the semigroups in $\mathbf{Li}$ by relying
on its equational description and a simple pumping argument.

\begin{toappendix}
  In preparation for the proof of \cref{sem_li_const}, we give the pumping lemma
  that we will use:

\begin{claim}
	\label{find_idem}
	Let $S$ be a semigroup of size $k$, and let $\omega$ be the idempotent
        power of~$S$. Let $\pi_S\colon S^+ \to S$ be the canonical morphism.
	Let $w$ be a word of length $k+1$ in $S^*$.
	There exists three words $w_0,w_1,w_2 \in S^*$ such that
        $w = w_0 w_1 w_2$ and
	$\pi_S(w) = \pi_S(w_0)\pi_S(w_1)^\omega\pi_S(w_2)$. 
\end{claim}

\begin{proof}
  Consider the images by~$\pi_S$ of all non-empty prefixes of~$w$. As $w$ has length $k+1$
  and there are $k+1$ non-empty prefixes but only $k$ different elements in~$S$,
  we know that two prefixes must have the same image. Formally, there are $1
  \leq i < j \leq k+1$ such that, letting $w_0$ be the prefix of length $i$
  of~$w$
  and defining $w_1$ such that $w_0 w_1$ is the prefix of length $j$ of~$w$
  and $w_2$ such that $w = w_0 w_1 w_2$, we
  have $\pi_S(w_0) = \pi_S(w_0 w_1)$. Hence, 
  $\pi_S(w_0) = \pi_S(w_0) \pi_S(w_1)$.
  Injecting the equation in itself $\omega-1$ times, we get:
  $\pi_S(w_0) = \pi_S(w_0) \pi_S(w_1)^\omega$.
  Together with the previous equation, we have:
  $\pi_S(w_0) \pi_S(w_1) = \pi_S(w_0) \pi_S(w_1)^\omega$.
  Thus, $\pi_S(w) = \pi_S(w_0) \pi_S(w_1) \pi_S(w_2) = \pi_S(w_0)
  \pi_S(w_1)^\omega \pi_S(w_2)$, which is what we wanted to show.
\end{proof}

We also justify the alternative equational characterization of
$\mathbf{Li}$ which was given without proof in the main text:

\begin{claim}[\cite{Pin2025Automata}]
  \label{pin}
	The following two equations generate the same variety of semigroups
        $\mathbf{Li}$:
	$ \llbracket x^\omega y x^\omega = x^\omega \rrbracket = 
	\llbracket x^\omega y z^\omega = x^\omega z^\omega \rrbracket$
\end{claim}

\begin{proof}
	Let $S$ be a semigroup.
Assume that for all $x,y \in S$ we have $x^\omega y x^\omega = x^\omega$. Then 
$x^\omega y z^\omega = x^\omega z^\omega x^\omega y z^\omega = x^\omega
z^\omega$, where the first equality is by replacing $x^\omega$ by $x^\omega
z^\omega x^\omega$ thanks to the equation that we assumed, and the second is by
replacing $z^\omega x^\omega y z^\omega$ by $z^\omega$ using the same equation.

Conversely, assume that for all $x,y,z \in S$ we have $x^\omega y z^\omega = x^\omega z^\omega$. 
Taking $z \coloneq x$, we obtain $x^\omega y x^\omega = x^\omega x^\omega = x^\omega$
which gives the first equation.
\end{proof}

We can now state and show \cref{sem_li_const}:
\end{toappendix}

\begin{lemmarep}\label{sem_li_const}
Fix a semigroup $S \in \mathbf{Li}$,
then its out-of-order evaluation problem has $O(1)$ space complexity.
\end{lemmarep}

\begin{proofsketch}
	Let $k = |S|$, store the first $k+1$ and last $k+1$ letters.
At the end return the product in order of all the stored elements. This
algorithm is in constant space.
	The correctness follows by using a pumping technique on the 
        first and last $k+1$
	elements to construct idempotent elements and then using the equation
	$\mathbf{Li}  = \llbracket x^\omega y z^\omega = x^\omega z^\omega\rrbracket$
	to conclude that the product of the stored elements is equal to the entire product.
\end{proofsketch}

\begin{proof}
Let $S$ be a finite semigroup, and let $k \coloneq |S|$.
Consider an out-of-order streaming sequence of length $n$ over $S$
for a word
$s = s[0] s[1] \cdots s[n-1]$.
The algorithm works as follows:
store the first $k+1$ and last $k+1$ letters.
At the end return the product in order of all the stored elements. This
algorithm is in constant space, all that remains is to show that it is correct.

The correctness claim is immediate if the word length is $2k+2$ or less.
Otherwise, 
writing $s = s_- s' s_+$ 
with $|s_-| = k+1$
and $|s_+| = k+1$ the prefix and suffix that we have stored,
we apply \cref{find_idem} on $s_-$ and on $s_+$
and find 
two idempotent elements $x^\omega$ and $y^\omega$ such that, writing
$\pi_S$ the canonical morphism,
we have $\pi_S(s_-) = x_- x^\omega x_+$ for some $x_-,x_+ \in S$ and likewise
we have $\pi_S(s_+) = y_- y^\omega y_+$ for some $y_-,y_+ \in S$. Hence,
$\pi_S(s) = \pi_S(s_-) \pi_S(s') \pi_S(s_+) = 
x_- x^\omega x_+
\pi_S(s')
y_- y^\omega y_+$. Now, we use the characterization of $\mathbf{Li}$ in
\cref{pin} as $\mathbf{Li}  = \llbracket x^\omega y z^\omega = x^\omega z^\omega
\rrbracket$ to argue that the right-hand-side is equal to 
$x_- x^\omega x_+ y_- y^\omega y_+$, i.e., it is the product $\pi_S(s_-)
\pi_S(s_+)$ of the stored prefix and suffix $s_-$ and $s_+$, which concludes.
%
\end{proof}

Together with the upper bound on commutative semigroups that follows from the
immediate generalization to semigroups of \cref{com_const}, and thanks to the
closure under semigroup operations by the analogue of \cref{closure_mono} for
semigroups, we deduce that semigroups in $\mathbf{Li} \lor \mathbf{Com}$ have
$O(1)$ space out-of-order semigroup evaluation. This establishes the upper
bound of~\cref{low_log_sem}. 


\subparagraph*{Lower bound.}
We now show that all other semigroups require logarithmic space, which completes
the proof of~\cref{low_log_sem} by showing the second point.

For this, we first give an 
equational characterization of the variety
$\mathbf{Li}\lor\mathbf{Com}$. 
This relates to the general study of semigroups 
of the form $\mathbf{V}\lor\mathbf{Li}$, with a general methodology to
understand these varieties given in~\cite{grosshans2021note}. 
We will prove the
equation $s^\omega xy t^\omega = s^\omega yx t^\omega$, and will in fact give an
alternative characterization via three equations. Each of these equations will
then be used to give a separate lower bound using fooling set constructions.

\begin{toappendix}
  In preparation for the proof of our characterization of $\mathbf{Li} \lor
  \mathbf{Com}$, we need to show that two equations define the same variety
  (which will in fact be $\mathbf{Li} \lor \mathbf{Com}$, as we will show):

\begin{lemmarep}
  \label{lem:addctx}
	The varieties of semigroups generated by these two equations are equivalent.
    $\llbracket s^\omega x y t^\omega = s^\omega y x t^\omega \rrbracket
    = \llbracket s^\omega\, u\, x y\, v \, t^\omega = s^\omega\, u \, y x\, v \, t^\omega \rrbracket$
\end{lemmarep}

\begin{proof}
	Assume 
	\begin{equation}
	\label{small}
		s^\omega\, x y\, t^\omega = s^\omega\, y x\, t^\omega
	\end{equation} then we have:
	\begin{align*}
		s^\omega\, u\, x y\, v 								\, t^\omega
		={}& s^\omega\, (s^\omega u)\, (x y v t^\omega) 		\, t^\omega & \text{by associativity and idempotence}\\
		={}& s^\omega\, (x y v t^\omega)\, (s^\omega u) 		\, t^\omega & \text{by \cref{small}}\\
		={}& s^\omega\, (x y) \, (v t^\omega) \, s^\omega \, u 		\, t^\omega & \text{by associativity}\\
		={}& s^\omega\,  (v t^\omega)\, (x y) \, s^\omega \, u 		\, t^\omega & \text{by \cref{small}}\\
		={}& s^\omega\, v \, t^\omega \, (x) (y) \, s^\omega \, u			\, t^\omega & \text{by associativity}\\
		={}& s^\omega\, v \, t^\omega \, (y) (x) \, s^\omega \, u			\, t^\omega & \text{by \cref{small}}\\
		={}& s^\omega\, (v t^\omega) \, (y x) \, s^\omega \, u			\, t^\omega  & \text{by associativity}\\
		={}& s^\omega\, (y x)\,(v t^\omega) \, s^\omega \, u			\, t^\omega  & \text{by \cref{small}}\\
		={}& s^\omega\, (y x v t^\omega) \, (s^\omega u) 		\, t^\omega  & \text{by associativity}\\
		={}& s^\omega\, (s^\omega u) \, (y x v t^\omega) 		\, t^\omega & \text{by \cref{small}}\\
		={}& s^\omega\, u \, y x\, v 								\, t^\omega & \text{by associativity and idempotence}
	\end{align*}

        Now assume 
		\begin{equation} \label{big}
			s^\omega\, u\, x y\, v\, t^\omega = s^\omega\, u\, y x\, v\, t^\omega
		\end{equation}
		Then we have:
	\begin{align*}
		s^\omega\, x y\, 								\,t^\omega
		={}& s^\omega\, s^\omega\, x y t^\omega		\, t^\omega &\text{by idempotence}\\
		={}& s^\omega\, s^\omega\, y x t^\omega		\, t^\omega &\text{by \cref{big} with $u := s^\omega$ and $v := t^\omega$}\\
                ={}& s^\omega\, y x \, t^\omega & & \hfill \qedhere\null
	\end{align*}
\end{proof}

We can now prove \cref{licom_eqs}:
\end{toappendix}

\begin{lemmarep}\label{licom_eqs}
  The following conditions are equivalent
   \begin{enumerate}
      \item The semigroup $S$ belongs to $\mathbf{Li}\lor \mathbf{Com}$.
      \item The semigroup $S$ satisfies the equation $s^\omega xy t^\omega = s^\omega yx t^\omega$.
      \item The semigroup $S$ satisfies the three equations: 
         \begin{align}
             s^\omega x s^\omega t^\omega &= s^\omega x t^\omega 
             \label{eq:licom1} \\
             s^\omega x s^\omega y s^\omega &= s^\omega x y s^\omega
             \label{eq:licom2} \\
             s^\omega x s^\omega y s^\omega &= s^\omega y s^\omega x s^\omega
             \label{eq:licom3}
         \end{align}
	\end{enumerate}
\end{lemmarep}
\begin{proofsketch}
        We can show the equivalence between points 2 and 3 simply by
        equation manipulation, so the key point is to show the equivalence of
        points 1 and 2.
	For the easy $1\to2$ direction we see that the equations for $\mathbf{Li}$ and $\mathbf{Com}$
	both imply the equation of point~2.
	For the $2\to1$ direction, similarly to \cref{fleq}, taking a
        semigroup $S$ satisfying the equation, we define
        one congruence of $\mathbf{Li}$ and one congruence of $\mathbf{Com}$ and
        we show that their intersection refines the congruence defined by~$S$.
        For this, we use pumping 
        to construct idempotents in the
        common prefixes and suffixes (thanks to the $\mathbf{Li}$ congruence),
        then using the equation between the idempotents to argue that all letters
        can be reordered and that we can conclude simply from the $\mathbf{Com}$
        congruence.
\end{proofsketch}
\begin{proof}
	We start by showing the equality between the first and second statements.

	We first show the easy $1\to2$ direction.
	Let $S \in \mathbf{Li}$.
	For all $s,x \in S$ we have $s^\omega x s^\omega = s^\omega$
	and thus for all $s,x,y$ we also have $s^\omega x s^\omega = s^\omega= s^\omega y s^\omega$.
	Using this equation we obtain
	\begin{align*}
		s^\omega xy t^\omega 
		={}& s^\omega t^\omega s^\omega xy t^\omega & \text{by} \quad s^\omega = s^\omega t^\omega s^\omega\\
		={}& s^\omega t^\omega s^\omega yx t^\omega 
		& \text{by} \quad t^\omega s^\omega yx t^\omega = t^\omega s^\omega xy t^\omega\\ 
		={}& s^\omega yx t^\omega & \text{by} \quad s^\omega = s^\omega t^\omega s^\omega
	\end{align*}
        Now, let $S \in \mathbf{Com}$. From $xy=yx$, we immediately get
        $s^\omega xy t^\omega = s^\omega yx t^\omega$. Hence, all
        semigroups in
        $\mathbf{Li} \lor \mathbf{Com}$ satisfy the equation of point 2.

	Now we show the $2\to1$ direction.
	Let $S \in \llbracket s^\omega x y t^\omega = s^\omega y x t^\omega \rrbracket$ be a semigroup, let $k$ be its size.
	We now show that the congruence defined by $S$ is refined by the
        intersection of the $\mathbf{Li}$ congruence $\sim_{Li,k+1}$ and a $\mathbf{Com}$ 
	congruence $\sim_{Com,k!,3\omega}$, noting that $\omega$ divides $k!$.
        More precisely, $\sim_{Li,k+1}$ equates any two words having the same
        $k+1$ first letters and $k+1$ last letters (in particular words of
        length less than $2k+2$ are only equivalent to themselves). It is
        immediate that this defines an equivalence relation, and that it is a
        congruence because when two words $s$ and $s'$ have the same $k+1$
        first and $k+1$ last letters, and likewise $t$ and $t'$, then the same
        is true of $st$ and $s't'$. As for $\sim_{Com,k!,\omega+2k+2}$, it equates any two
        words $u, v$ of~$S^*$ such that, for each element $x \in S$, the
        respective number of
        occurrences $i$ and $j$ of $x$ in $u$ and in $v$ ensure that $i$ and $j$
        have the same remainder modulo $k!$ and that we have
        $\min(i,\omega+2k+2) =
        \min(j,\omega+2k+2)$:
        again this is clearly a congruence. (Note that, in contrast with the
        $RED_{t,p}$ congruences from \cite{tesson} used in the proof of
        \cref{fleq}, our congruence now accounts for a threshold in addition to
        a modulo. This is still commutative, and the reason why it was not
        needed in the proof of \cref{fleq} is that thresholds are intuitively
        handled already by $\mathbf{FL}$, whereas the same is not true of
        $\mathbf{Li}$.)

        So take two words $u$ and $v$ that are equivalent for
        $\sim_{Li,k+1}$ and for $\sim_{Com,k!,\omega+2k+2}$. If one of $u$ and $v$ has
        length $<2k+1$, then $u$ and $v$ must be identical (thanks to the first
        congruence) and then they obviously have the same product in~$S$ so
        there is nothing to show. Now, if $u$ and $v$ both have length $\geq
        2k+2$, thanks to the first congruence
        we can write $u = s u' t$ and $v = s v' t$  with $s, t$ having
        length $k+1$. 
        Using \cref{find_idem} like in the proof of \cref{sem_li_const}, 
        letting $\pi_S \colon S^+ \to S$ be the canonical morphism,
        we can write $u = s_- s' s_+ u' t_- t' t_+$ and 
        $v = s_- s' s_+ v' t_- t' t_+$ such that $\pi_S(u) = \pi_S(s_-)
        \pi_S(s')^\omega \pi_S(s_+)
        \pi_S(u') \pi_S(t_-) \pi_S(t')^\omega \pi_S(t_+)$ and  likewise
        $\pi_S(v) = \pi_S(s_-) \pi_S(s')^\omega \pi_S(s_+) \pi_S(v') \pi_S(t_-)
        \pi_S(t')^\omega \pi_S(t_+)$.

        The equation assumed in point~2 is equivalent, by \cref{lem:addctx}, to:
        \[s^\omega\, u\, x y\, v \, t^\omega = s^\omega\, u \, y x\, v \,
        t^\omega\]
        Thus, in the expressions for $\pi_S(u)$ and $\pi_S(v)$ given above,
        between $\pi_S(s')^\omega$ and $\pi_S(t')^\omega$, we can commute
        everything without changing the product in~$S$. For this reason, we can
        substitute $\pi_S(u')$ by $\pi_S(u'')$ for $u'' \coloneq \prod_{m \in S} m^{i_m}$
        and substitute $\pi_S(v')$ by $\pi_S(v'')$ for $v''
        \coloneq \prod_{m \in S} m^{j_m}$. Further, we know that $u$ and $v$ are
        $\sim_{Com,k!,\omega+2k+2}$-equivalent, so the same is true of $u'$ and $v'$ (up
        to eliminating the common prefix and suffix). This implies that, for
        each $m \in S$, the values $i_m$ and $j_m$ have the same remainder
        modulo $k!$. Further, from the threshold information, if $i_m < \omega$
        or $j_m < \omega$ then we must have $i_m = j_m$. This uses the fact that
        there are at most $2k+2$ occurrences of~$m$ in $u$ outside of $u'$, and
        in $v$ outside of~$v'$. So the only point to argue is that we can
        replace the $i_m$ by $j_m$ when $i_m \geq \omega$ and $j_m \geq \omega$
        in~$u''$ without changing the product in~$S$ of $\pi_S(u)$: this is
        simply because $m^{i_m} = m^\omega m^{i_m - \omega} = m^\omega m^{i_m
        \bmod \omega} = m^\omega m^{j_m \bmod \omega} = m^{j_m}$. Doing this for
        all letters, we finally conclude that indeed $\pi_S(u) = \pi_S(v)$. So
        we have established that the intersection of the two congruences
        $\sim_{Li,k+1}$ and $\sim_{Com,k!,\omega+2k+2}$ indeed refines the
        congruence defined by~$S$. Hence, the latter congruence has equivalence
        classes which are unions of classes of the intersection; by similar
        reasoning to the proof of \cref{fleq} we deduce that $S$ is a quotient
        of the product of a semigroup of~$\mathbf{Li}$ and a semigroup
        of~$\mathbf{Com}$.

        \medskip

	We then show the equality between point 2 and point 3.

	For the easy $2\to3$ direction, let $S$ be a semigroup.
	Assume $S \in \llbracket s^\omega x y t^\omega = s^\omega y x t^\omega
        \rrbracket$.
	Then:
	\begin{enumerate}
	    \item Let $y \coloneq s^\omega$ we have $s^\omega x s^\omega t^\omega = 
		    s^\omega s^\omega x t^\omega = s^\omega x t^\omega$
		which is equation \eqref{eq:licom1}.
		  
	    \item Let $t \coloneq s^\omega$ we have $s^\omega x s^\omega y s^\omega 
		    = s^\omega y x s^\omega s^\omega = s^\omega x y s^\omega$
		which is equation \eqref{eq:licom2}.

	    \item Let $t \coloneq x$ we have $s^\omega x s^\omega y s^\omega 
		    = s^\omega s^\omega x s^\omega y s^\omega = s^\omega y s^\omega x s^\omega$
		which is equation \eqref{eq:licom3}.
	\end{enumerate}

	For the $3\to2$ direction,
	assume that $S$ satisfies all three equations
        \eqref{eq:licom1}--\eqref{eq:licom3}. We show that the equation of
        point~2 holds. Indeed,
	\begin{align*}
	s^\omega x y t^\omega 
	  &= s^\omega x y s^\omega t^\omega & \text{by \eqref{eq:licom1}} \\
	  &= s^\omega x s^\omega y s^\omega t^\omega & \text{by \eqref{eq:licom2}} \\
	  &= s^\omega y s^\omega x s^\omega t^\omega & \text{by \eqref{eq:licom3}} \\
	  &= s^\omega y x s^\omega t^\omega & \text{by \eqref{eq:licom2}} \\
          &= s^\omega y x t^\omega & \text{by \eqref{eq:licom1}} & \null \hfill \qedhere
	\end{align*}
\end{proof}

The lemma above gives an equational characterization of $\mathbf{Li} \lor
\mathbf{Com}$ with the equation $s^\omega xy t^\omega = s^\omega yx t^\omega$,
intuitively implying that everything commutes between two (possibly different) idempotents. However, to show the lower bound part of \cref{low_log_sem},
we will use the three equations of the third point. Namely, for any semigroup
not in $\mathbf{Li} \lor \mathbf{Com}$, we will see which equation among the
three is violated, and use a different proof in each case.

We first notice that the case of a
violation of the third equation $s^\omega x s^\omega y s^\omega = s^\omega y
s^\omega x s^\omega$ in fact precisely corresponds to the case of a
non-commutative semigroup, and the corresponding lower bound was already proven
as
\cref{lem:lcom}. Hence, only the first and second equations remain. 
We give the construction for the first equation:

%
\begin{lemmarep}\label{st_swap}
	For any fixed semigroup $S \not\in \llbracket s^\omega x s^\omega t^\omega = s^\omega x t^\omega \rrbracket$,
	the out-of-order evaluation problem for~$S$ has $\Omega(\log(n))$ space complexity.
\end{lemmarep}

\begin{proofsketch}
	We invoke \cref{thm:fooling_lower_bound} using the following
        $(n-2,2n)$-fooling set for any $n > 0$.
	For each pair of indices $i,j \in \{2, \dots, n-1\}$ define 
	$$w_i = (s^\omega \ept)^i \, (t^\omega \ept)^{\,n-i}
        \quad \text{and} \quad v_{ij} = (\ept s^\omega)^{i-1} \, (\ept x) \, (\ept t^\omega)^{\,n-i}$$
	Thus the out-of-order evaluation problem for $S$ has $\Omega(\log(n))$ complexity.
\end{proofsketch}

\begin{proof}
	Let $S$ be a finite semigroup such that there exist $s,x,t \in S$ with 
	$s^\omega x s^\omega t^\omega \neq s^\omega x t^\omega$.
	Let $n \in \mathbb N$. 
	Define the following $(n-2,2n)$-fooling set.
	For $i \in \{2, \dots, n-1\}$ define 
	$w_i \coloneq (s^\omega \ept)^i \, (t^\omega \ept)^{\,n-i}$.
	For any two indices $i<j$, define the \textit{complementary witness} 
	$v_{ij} \coloneq (\ept s^\omega)^{i-1} \, (\ept x) \, (\ept t^\omega)^{\,n-i}$.
	This satisfies $w_i \circ v_{ij} = s^\omega x t^\omega \not = 
	s^\omega x s^\omega t^\omega = w_j \circ v_{ij}$.
	By \cref{thm:fooling_lower_bound},
	the out-of-order evaluation problem for $S$ has $\Omega(\log(n))$ complexity.
\end{proof}

\begin{example}
	\label{ex:sg_lower_bound}
	The semigroup $S(a^*bc^*)$ has four elements $\{a, b, c, 0\}$, with $a$, $c$, and $0$ being idempotent.  It does
			not satisfy the equation $s^\omega x s^\omega
			t^\omega = s^\omega x t^\omega$ when taking $s=a$, $x=b$,
			$t=c$: the left-hand side evaluates to $abac = 0$
			and the right-hand side evaluates to $abc = b$. Hence,
			its out-of-order evaluation problem requires
			$\Omega(\log n)$ space by~\cref{st_swap}.  
		
			Note that we can also devise an $O(\log n)$ space
                        algorithm for
			this semigroup by storing the maximal occurrence of $a$,
			the minimal occurrence of $c$,
                        the minimal occurrence of~$b$,
                        and the number of occurrences of
			$0$ and $b$.
                        The result is~$0$ if there is a~$0$ or if there are two
                        or more~$b$. If there is no $b$ the result is $a$ or $c$
                        or $0$ depending on whether $a$ or $c$ or both are
                        present. Last, if there is one $b$, the result is $b$
                        unless the rightmost $a$ is right of the~$b$ or the
                        leftmost $c$ is left of the~$b$, in which case the
                        result is~$0$.
                        Thus, the out-of-order evaluation problem
                        for this semigroup is in $\Theta(\log n)$.
\end{example}

We then give the construction for the second equation:

\begin{lemmarep}\label{xy_sep}
	For any fixed semigroup $S \not\in \llbracket s^\omega x s^\omega y s^\omega = s^\omega x y s^\omega \rrbracket$, 
	the out-of-order evaluation problem for~$S$ has $\Omega(\log(n))$ space complexity.
\end{lemmarep}

\begin{proofsketch}
	We invoke \cref{thm:fooling_lower_bound} using the following $(n-2,2n)$-fooling set 
        for any $n>0$.
	For each pair of indices $i,j \in \{2, \dots, n-1\}$ define 
	$$w_i \coloneq (s^\omega \ept)^{i-1} (x \ept) (s^\omega \ept)^{\,n-i}
        \quad \text{and} \quad v_{ij} = (\ept s^\omega)^{j-1} \, (\ept y) \, (\ept s^\omega)^{\,n-j}$$
	Thus the out-of-order evaluation problem for $S$ has $\Omega(\log(n))$ space complexity.
\end{proofsketch}

\begin{proof}
	Let $S$ be a finite semigroup such that there exist $s,x,y \in S$ with 
	$s^\omega x s^\omega y s^\omega \neq s^\omega x y s^\omega$.
	Let $n \in \mathbb N$. 
	Define the following $(n-2,2n)$-fooling set.
	For $i \in \{2, \dots, n-1\}$ define 
	$w_i = (s^\omega \ept)^{i-1} (x \ept) (s^\omega \ept)^{\,n-i}$.
	For any two indices $i<j$, define the \textit{complementary witness} 
	$v_{ij} = (\ept s^\omega)^{j-1} \, (\ept y) \, (\ept s^\omega)^{\,n-j}$.
	This satisfies $w_i \circ v_{ij} = s^\omega x s^\omega y s^\omega \not = 
	s^\omega x y s^\omega = w_j \circ v_{ij}$.
	By \cref{thm:fooling_lower_bound},
	the out-of-order membership problem therefore has $\Omega(\log(n))$ space complexity.
\end{proof}


\begin{example}
	\label{ex:sg_lower_bound2}
	The semigroup $S(a^*bba^*)$ has elements $\{a, b, ab, ba,
		bb, 0\}$ and does not satisfy the equation $s^\omega
		x s^\omega y s^\omega = s^\omega x y s^\omega$ when taking
		$s=a$, $x=y=b$: the left-hand side evaluates to
		$ababa = 0$ and the right-hand side to $abba = bb$.
		Hence, its out-of-order evaluation
		problem requires $\Omega(\log n)$ space by
		\cref{xy_sep}.
	
		Note that this is the same underlying language as
		in~\cref{ex:abba}, but considering the syntactic
		semigroup instead of the syntactic monoid. 
		We can also devise an $O(\log n)$ space
		algorithm for this semigroup 
		by storing at
		most two positions for each of the elements $ab$, $ba$,
		$bb$, and $0$. If any of these elements occurs more
		than twice, we can immediately return $0$. Otherwise,
		we know that all positions that we have not stored are
		labeled by $a$, which gives us
		enough information to compute the product.
		Thus, the out-of-order evaluation problem
		for this semigroup is in $\Theta(\log n)$.
\end{example}

We can now conclude the proof of \cref{low_log_sem}: by combining \cref{licom_eqs}, \cref{st_swap}, \cref{xy_sep}, and \cref{lem:lcom}, 
we conclude that any semigroup $S \not\in \mathbf{Li} \lor \mathbf{Com}$ 
must have $\Omega(\log(n))$ space complexity.

We have illustrated that there are novel algorithms for the case of
semigroups which did not exist in the case of monoids.
Yet, the case of out-of-order evaluation for semigroups
is still coarser than out-of-order
membership problem 
for regular languages.
For instance, we can show the following linear semigroup lower bound, which is
in contrast to the $O(1)$ space upper bound for the corresponding language in~\cref{ex:abstar}.

\begin{propositionrep}\label{ex:sg_lower_bound_ab}
	The out-of-order evaluation problem for $S((ab)^*)$ has $\Theta(n)$
        space complexity.
\end{propositionrep}
\begin{proof}
	The upper bound is a direct consequence of~\cref{word_upper_n}.
	To prove the lower bound we will define a fooling set,
	and conclude using \cref{thm:fooling_lower_bound}.

	We build a $(2^n,2n)$-fooling set for any $n>0$.
	For each $\alpha\in\{0,1\}^n$ we define $w_{\alpha,i}$ by either
        $a \ept$ or $(ab) \ept$ depending on the $i$-th bit of $\alpha$. (Here, recall
        that crucially, $ab$ is an element of the semigroup, i.e., it is a
      single element and we always have $|w_{\alpha,i}| = 2$.
        Now, let
	$w_\alpha \coloneq w_{\alpha,1} w_{\alpha,2} \dots w_{\alpha,n}$,
        and let $W_n \coloneq \{w_\alpha \mid \alpha \in \{0,1\}^n\}$ be
        the fooling set.

	For any $\alpha \neq \beta$, we build the complementary witness
	similarly. Namely, we define $v_{\alpha,i}$ as an element
        $\ept (bab)$ or $\ept (ab)$ depending on the $i$-th bit of~$\alpha$: 
        this ensures that $w_{\alpha,i} \circ v_{\alpha,i} =
        abab$.
        Now, we let $v_{\alpha} \coloneq v_{\alpha,1} v_{\alpha,2} \dots
        v_{\alpha,n}$, and we simply let $v_{\alpha,\beta} \coloneq v_\alpha$.
	The construction then ensures that $w_\alpha \circ v_\alpha = (abab)^n$ 
	which corresponds to the semigroup element $ab$. Now, 
	let $\iota(\alpha,\beta)$ be the first position where $w_\beta$ and $w_\alpha$ differ.
	Thus $w_{\beta,\iota(\alpha,\beta)} \circ v_{\alpha,\iota(\alpha,\beta)}$ will be either $aab$ or $abbab$
	which are both the $0$ element in the semigroup and thus $w_\beta \circ
        v_{\alpha,\beta} = 0$.
	This implies that the out-of-order evaluation problem for $S((ab)^*)$ 
	has $\Omega(n)$ space complexity.
\end{proof}

\section{Other Complexity Regimes}
\label{sec:horror}
Our results in the previous section have characterized the semigroups 
enjoying constant-space out-of-order semigroup evaluation; 
with a logarithmic lower bound applying to all other semigroups. 
However, unlike our result for monoids (\cref{thm:monoid} in \cref{sec:monoid}),
we have not given a classification of all possible complexity regimes.

Unfortunately, in our present understanding this appears quite challenging, as
we will now illustrate. For convenience the results presented in this section 
are stated for out-of-order membership to languages, but we believe that the phenomena illustrated here already apply to
out-of-order evaluation for semigroups (beyond the constant-space regime).
Thus, classifying out-of-order semigroup evaluation, or out-of-order membership to
regular languages, appears much more challenging than the monoid case.

We first show that some languages admit unexpected logarithmic-space
algorithms, and then show an example of a language for which we
can devise an $O(\sqrt{n})$-space algorithm.


\subparagraph*{Logarithmic-space examples.}
The first example that we present is the 
language $L = a^* b^* a^*$. First note that the syntactic monoid of~$L$ is not in $\mathbf{FL} \lor \mathbf{Com}$: intuitively, there is no $k
> 0$ for which we can solve out-of-order membership (or evaluation) by
remembering the $k$ first and last occurrences of each letter along with
commutative information. More precisely, the first and last $a$'s are
uninformative, and the first and last $b$'s do not allow us to know whether
there could be some $a$'s surrounded by two $b$'s. So, the out-of-order monoid
evaluation problem for the syntactic monoid of~$L$ has linear space complexity
by \cref{thm:monoid}. 
%
%
By contrast, we can show that out-of-order membership for~$L$ can be done in logarithmic space, with an ad-hoc algorithm:

\begin{propositionrep}
  \label{prp:aba}
	The out-of-order membership problem for $L = a^* b^* a^*$ has $\Theta(\log n)$ space complexity.
\end{propositionrep}

\begin{proofsketch}
  We memorize the minimal position $b_{\min}$ and maximal position $b_{\max}$
  where a $b$ occurs, along with the number
  of~$b$'s. A word belongs to the language precisely when the $b$'s form a
  contiguous interval, which amounts to verifying whether the number of $b$'s is
  equal to the difference between the minimal and maximal positions of~$b$ (plus 1).

	Now, for the lower bound, we construct the following $(n-1,2n)$-fooling set for any $n \in \mathbb N$. 
        For any $1 \leq i \leq n-1$, define
	$w_i \coloneq (a \ept)^i \, (b \ept)^{\,n-i}$.
	Now, for any $i,j \in \{1, \dots, n-1\}$ with $i<j$ define 
        $v_{ij} \coloneq (\ept a)^{j} \, (\ept b)^{\,n-j}$.
	This ensures that $w_i \circ v_{ij}$ is in $a^* (ba)^+ b^+$, i.e., it is
        not in~$L$; whereas $w_j \circ v_{ij}$ is in $a^* b^+$, i.e., it is
        in~$L$.
	By \cref{thm:fooling_lower_bound}, the out-of-order membership problem
        for $a^*b^*a^*$ then has $\Omega(\log(n))$ complexity.
\end{proofsketch}

\begin{proof}
	All we need to check is that the $b$'s form a contiguous interval.
	To do this we compute the following information, which takes
        logarithmic space:
        \begin{itemize}
          \item The minimal position $b_{\min}$ where a $b$ occurs;
          \item The maximal position $b_{\max}$ where a $b$ occurs;
          \item The total number $n_b$ of $b$'s.
        \end{itemize}

	At the end we accept or reject depending on whether the following
        equality holds:
        \[
          b_{\max} - b_{\min} = n_b-1.
        \]
        It is easy to see that this holds precisely when the $b$'s form a
        contiguous interval, concluding the proof.
	
	Now, for the lower bound, we proceed using fooling sets. The argument
        was entirely given in the proof sketch, but we repeat here verbatim for
        convenience. We construct the following $(n-1,2n)$-fooling set for any $n \in \mathbb N$. 
        For any $1 \leq i \leq n-1$, define
	$w_i \coloneq (a \ept)^i \, (b \ept)^{\,n-i}$.
	Now, for any $i,j \in \{1, \dots, n-1\}$ with $i<j$ define 
        $v_{ij} \coloneq (\ept a)^{j} \, (\ept b)^{\,n-j}$.
	This ensures that $w_i \circ v_{ij}$ is in $a^* (ba)^+ b^+$, i.e., it is
        not in~$L$; whereas $w_j \circ v_{ij}$ is in $a^* b^+$, i.e., it is
        in~$L$.
	By \cref{thm:fooling_lower_bound}, the out-of-order membership problem
        for $a^*b^*a^*$ then has $\Omega(\log(n))$ complexity.
\end{proof}

Note that this algorithm does not work for the syntactic monoid of $a^* b^*
a^*$, because we cannot count the neutral elements 
between $b_{\min}$ and $b_{\max}$, rendering the technique useless.
The result generalizes to languages such as $a^* b^* c^* b^* a^*$, by
maintaining the counts and minima and maxima of each letter and checking first
that the $c$'s form a contiguous interval, and then that the $b$'s form a
contiguous interval around the $c$'s (offsetting the count of the $b$'s by that
of the $c$'s).
%

One can wonder whether the result also generalizes to languages of the form $a^*
b^* a^* b^* a^*$. This is not immediate, because now we cannot easily identify
the $a$'s in the middle interval, or the $b$'s in the left interval versus the
right interval. However, it turns out that logarithmic space complexity can be
shown using a more intricate approach (which also could have been used to show
\cref{prp:aba}):

\begin{propositionrep}
  \label{prp:ababa}
	The out-of-order membership problem for $a^* b^* a^* b^* a^*$
        has $O(\log n)$ space complexity.
\end{propositionrep}

\begin{proofsketch}
  We keep track of the first and last $b$'s along with the number $n_a$ of $a$'s, the
  sum $p_a$ of the positions featuring $a$'s, and the sum $q_a$ of the square of positions
  featuring~$a$'s. At the end of the algorithm, we can fix $n_a, p_a, q_a$ to
  subtract all positions left of the first $b$ or right of the last $b$ -- we
  know that these positions must contain $a$'s. We let $n_a', p_a', q_a'$ be the
  resulting quantities. Then, we can check that the
  remaining $a$'s form a contiguous interval within the $b$'s by using their
  number $n_a'$ and their sum $p_a'$
  to find a candidate interval, and then verifying that the sum of squares
  $q_a'$ is
  correct. For this, we use the fact that placing $n_a'$ integers of range
  $p_a'$ as a contiguous interval, when possible, is the unique way to minimize
  the sum of squares $q_a'$ proved in \cref{sumof}.
\end{proofsketch}

\begin{toappendix}
  In preparation for the proof of \cref{prp:ababa}, we will need a lemma about
  how intervals can be uniquely identified from their cardinality and sum thanks
  to the fact that they minimize the sum of squares:
\begin{lemma}
	\label{sumof}
        Let $M > 0$,
consider the interval $I=\{p,p+1,\dots,p+n-1\}\subseteq\{1,\dots,M\}$ that
starts at~$p$ and contains $n$ elements.
Let $X\subseteq\{1,\dots,M\}$ be a set of $n$ elements (i.e., $|X|=n$) 
and ensuring $\sum_{x\in X}x=\sum_{y\in I}y$. We then have $\sum_{x\in X}x^2\ge\sum_{y\in I}y^2$, with equality iff $X=I$.
\end{lemma}

\begin{proof}
  The backward direction of the equivalence is immediate, so we prove the
  forward direction.

Write $X=\{x_1<\dots<x_n\}$. We consider holes of the interval $[x_1,x_n]$,
i.e., the integers $x_1 < h < x_n$ which are not in~$X$), and distinguish
three cases depending on how many holes exist.

\emph{Case 1 (no holes).} If every integer $k$ with $x_1\le k\le x_n$ belongs to
$X$ then $X=\{a,a+1,\dots,a+n-1\}$ for some $a$; the hypothesis $\sum_{x\in
X}x=\sum_{y\in I}y$ then forces $a=p$, so $X=I$.

\emph{Case 2 (exactly one hole).} Suppose exactly one integer $x_1 < h <
x_n$ is not in~$X$. If $x_1\ge p$ then we have $x_i\ge p+i-1$ for all $i$, and
since we skip at least one value we must have $x_n>p+n-1$, hence $\sum_{i=1}^n
x_i>\sum_{y\in I}y$, contradiction. If $x_1<p$ then we have $x_i\le p+i-1$ for
all $i$ because we skip at most one value, so with $x_1 < p$ we obtain $\sum_{i=1}^n x_i<\sum_{y\in
I}y$, a contradiction. Thus a one-hole configuration is impossible. (Note this impossibility uses only the equality of the sums and so holds at every stage below, since our operations preserve the total sum.)

\emph{Case 3 (at least two holes).} Let $x_1 < h<h' < x_n$ be two missing
integers in $(x_1,x_n)$. It is easy to see that there must exist elements $u$ and $v$ of~$X$ which are respectively before $h$
and after $h'$, and are respectively immediately at the left and immediately at
the right of a hole. 
Formally, choose $u\in X$ maximal with $u<h$ and $v\in X$ minimal with $v>h'$,
and set $u'\coloneq u+1$, $v'\coloneq v-1$: by assumption we have $u' \notin X$ and $v' \notin X$.
Now, rewrite the set $X$ by defining $X'=(X\setminus\{u,v\})\cup\{u',v'\}$.
Then $X'$ has the same cardinality as $X$, and it has the same total sum because
one element was incremented and one element was decremented. 
Moreover, the sum of squares of $X'$ minus that of $X$ is equal to:
\[
(u')^2+(v')^2-(u^2+v^2)=(u+1)^2+(v-1)^2-(u^2+v^2)=2(u-v)+2.\]
Since $u\le h-1$ and $v\ge h'+1\ge h+2$ because we have at least two holes, we have $u-v\le-3$, so $2(u-v)+2\le-4<0$. Thus the sum of squares strictly decreases when passing from $X$ to $X'$.

Now iterate the operation of Case 3: each step preserves cardinality and total sum, and strictly decreases the integer $\sum_{x\in X}x^2$, so the process terminates after finitely many steps at a set $X^{(\ell)}$ to which Case 3 no longer applies. Now, by the remark after Case 2, Case 2 cannot occur at any intermediate stage (nor at the terminal stage) because the total sum equals that of $I$ throughout the rewriting. Hence the terminal stage must be Case 1, so $X^{(\ell)}=\{a,a+1,\dots,a+n-1\}$ and the total-sum equality forces $a=p$, therefore $X^{(\ell)}=I$.

Since every Case-3 step strictly decreased $\sum_{x\in X}x^2$ and the terminal set is $I$, we conclude $\sum_{x\in X}x^2\ge\sum_{y\in I}y^2$, with equality precisely when no Case-3 step is needed, i.e. when $X=I$. This completes the proof.
\end{proof}

With this tool in hand, we are now ready to prove \cref{prp:ababa}

\begin{proof}[Proof of \cref{prp:ababa}]
	All we need to check is that the $a$ which are between the first and
        last $b$ form a contiguous interval. For this, we store the following
        information in logarithmic space:
	\begin{itemize}
            \item $b_{\min}$: the smallest position containing a $b$
            \item $b_{\max}$: the biggest position containing a $b$
            \item $n_a$: the number of $a$'s
            \item $p_a$: sum of the positions where $a$ occurs
            \item $q_a$: the sum of the square of the positions where $a$
              occurs
	\end{itemize}
        At the end of the algorithm, we know that every position before
        $b_{\min}$, and after $b_{\max}$, must contain an~$a$. This means that:
        \begin{itemize}
          \item The total number of $a$ between $b_{\min}$ and $b_{\max}$ is:
            \[ n_a'
            \coloneq n_a - (b_{\min}-1) - (n-b_{\max})\]
          \item The sum of the positions where $a$ occurs and which are
            between $b_{\min}$ and $b_{\max}$ are
            \[
              p_a' \coloneq p_a - \left(\sum_{i=1}^{b_{\min}-1} i\right)
              - \left(\sum_{i=b_{\max}+1}^{n} i \right) 
            \]
          \item The sum of the squares of the positions where $a$ occurs and which are
            between $b_{\min}$ and $b_{\max}$ are
            \[
            q_a' \coloneq q_a - \left(\sum_{i=1}^{b_{\min}-1} i^2\right)
            - \left(\sum_{i=b_{\max}+1}^{n} i^2 \right) 
            \]
        \end{itemize}
        Thus, between $b_{\min}$ and $b_{\max}$ we have $n_a'$ occurrences of~$a$,
        whose sum is $p_a'$ and whose sum of squares is $q_a'$. From $p_a'$, we
        deduce a candidate interval $[l,r]$ of positions of~$a$, with 
        $r-l+1 = n_a'$ and with $\sum_{i=l}^r i = p_a'$. It is easy to see that
        there is at most one such interval; if no such interval gives the right
        sum, then we reject.

        Then, we control that the candidate interval is correct by checking that
	$\sum_{i=l}^r i^2 = q_a'$. Thus using \cref{sumof} we obtain that this
        equality is true precisely when the set of positions of~$a$ between
        $b_{\min}$ and $b_{\max}$ is exactly the candidate interval, i.e.,
        precisely when the innermost $a$'s are contiguous. So the algorithm can
        reply accordingly and it correctly reflects whether the word belongs to
        the target language.
%
%
\end{proof}
\end{toappendix}

For the same reasons as previously stated, this argument does not work for
monoids because of the neutral element. It appears that the algorithm could be
combined with that of \cref{prp:aba} to show a logarithmic space upper bound for
more complex languages, e.g., $a^*b^*a^*c^*a^*b^*a^*$. However it
does not appear easy to generalize it to a higher number of alternations between
the same symbol, e.g., $a^*b^*a^*b^*a^*b^*$.

\subparagraph*{An $O(\sqrt{n})$-space algorithm.}
We now conclude the section by showing that, for the language $a^*b^*a^*b^*a^*b^*$, we can
nevertheless achieve sublinear space complexity:


\begin{propositionrep}
  \label{prp:ababab}
	The out-of-order evaluation problem for $a^* b^* a^* b^* a^* b^*$ has $O(\sqrt{n})$ space complexity.
\end{propositionrep}

\begin{proofsketch}
	We split the word $w$ of length $n$ into $\lceil \sqrt{n} \rceil$ 
	blocks of size $b = \lceil \sqrt{n} \rceil$, each either \emph{fully
        memorized} (i.e., we have an explicit table of its contents) or summarized by the first letter that is seen in the block. 
	When reading letter $x$ at position $i$ in block $B$, we update $B$ as follows: if $B$ is fully memorized, 
	we write $x$ in its table; if $B$ is empty, it becomes an $x$-block; if $B$ is an $x$-block, do nothing; if 
	$B$ is a $y$-block with $y\neq x$, we fully memorize $B$ and write $x$
        in its table. Whenever we reach more than 
	6 fully memorized blocks, we \emph{reject early}, guaranteeing $w \notin L$ since $w$ contains a 
	scattered subword of the form $(ab|ba)^6$, which is impossible in $L$. At the end, we reconstruct a 
	word $w'$ by completing all blocks: the block contents are $x^b$ for
        non-fully-memorized $x$-blocks, and the table contents for fully 
	memorized blocks, completing the missing cells by the first letter seen
        in the block to compensate for the letters streamed before we started
        fully memorizing the block.
        We can show that $w' = w$ so that testing $w' \in L$ is correct. The
        space usage is 
	$O(\sqrt{n})$ since we store at most 6 full blocks plus block types.
\end{proofsketch}

\begin{proof}
  Let $L = a^* b^* a^* b^* a^* b^*$ be the target language.

	Knowing the length $n$ of the word~$w$,
        we split $w$ in $\lceil \sqrt{n} \rceil$
        blocks of size $b = \lceil \sqrt{n} \rceil$.
        Each block is either \emph{fully memorized} or not: for each fully
        memorized block we have a table of size $b$ storing its explicit
        contents. (We will make sure throughout the algorithm that there are at
        most 6 fully memorized blocks.)
        Further, for each block
        we remember the first letter (if any) that we have seen in the block,
        i.e., we have empty blocks, $a$-blocks, and $b$-blocks.

        When a letter $x \in \Sigma$ arrives at position~$i$, we compute the
        block in which the position $i$ falls, the offset $j$ of the letter in
        its block, and we do the following:
        \begin{itemize}
          \item If the $i$-th block is fully memorized, then we have a table
            $T_i$
            storing the explicit contents of the block, and we write down $T[j]
            \coloneq x$.
          \item Otherwise, if the $i$-th block is an empty block, then it becomes an $x$-block
          \item Otherwise, if the $i$-th block is an $x$-block, then we do
            nothing
          \item Otherwise, if the $i$-th block is a $y$-block
            with $y \neq x$, then the block becomes fully memorized, we
            initialize a table $T_i$ to store its explicit contents, and we
            write down $T[j] \coloneq x$.
          \item Last, if there are now more than 6 fully memorized block, then
            we \emph{reject early} and determine that $w \notin L$. (More accurately, we commit
            to rejecting at the end of the stream, disregard all other tuples,
            and reject at the end.)
        \end{itemize}

        At the end of the algorithm, we do the following. Note that no block can
        be empty. We define a word $w'$ obtained by filling all missing cells of
        fully memorized $x$-blocks for $x\in \Sigma$ by $x$, and reading the
        blocks in sequence, using the completed table for the contents of the
        fully memorized blocks, and reading the word $x^b$ for the
        non-fully-memorized $x$-blocks. Note that no blocks are empty at this
        stage. We then test $w'$ for membership to~$L$ and answer accordingly.
        (Note that we can in fact test whether $w' \in L$ in streaming without
        having to materialize $w'$, i.e., the working memory required at the end
        is no worse than $O(\sqrt{n})$.)

        It is clear that the algorithm is indeed in $O(\sqrt{n})$, because it
        stores the type of each block in $O(\sqrt{n})$ overall, and stores at
        most a constant number of fully memorized blocks, each taking space
        $O(\sqrt{n})$.

        We must now explain why the algorithm is correct. We first show that,
        when the algorithm rejects early, then indeed $w \notin L$. But early
        rejection ensures that the word $w$ is of the form $\Sigma^* B_1
        \cdots \Sigma^* B_6 \Sigma^*$, with $B_1, \ldots, B_6$ the fully
        materialized blocks; and for a block to become fully materialized we
        must have seen both an $a$ and a $b$ in the block. So $w$ must
        contain a scattered subword $z$ of the form $(ab|ba)^6$, and it can be
        checked that no word of~$L$ can contain $z$ as a scattered subword.
        This witnesses that the number of $a$-to-$b$ and $b$-to-$a$ transitions
        in~$w$ is
        at least 6, whereas the words of $L$ have precisely five such
        transitions.

        Second, we show that when the algorithm does not reject early then it returns
        the right answer. We show this by establishing that the word $w'$
        constructed at the end of the algorithm is in fact equal to~$w$. For
        this, we reason block by block, according to the block types:
        \begin{itemize}
          \item For $x$-blocks with $x \in \Sigma$, we know that we have only
            seen $x$ in such blocks, so indeed they consist of $x^b$
          \item For fully memorized $x$-blocks with $x\in \Sigma$,
            we have written down all letters in such blocks from 
            the moment where they were fully memorized. So indeed all missing
          positions correspond to letters that were streamed before the block
          was fully memorized. So these letters must be~$x$, otherwise the block
          would have become fully memorized earlier. Thus, our reconstruction of
          the contents of such blocks is correct.
        \end{itemize}

        This establishes that the algorithm is correct, and concludes the proof.
\end{proof}

This algorithm would generalize to more complex languages, e.g., $(a^* b^*)^k$
and $(a^* b^*)^k a^*$ for any constant~$k$. We do not know whether the space
complexity achieved by this algorithm is optimal: it would be very interesting
to show that such languages admit an $\Omega(\sqrt{n})$ lower bound, although at
this stage we cannot rule out the existence of an $O(\log n)$ algorithm (or
indeed other sublinear upper bounds below $O(\sqrt{n})$).

\section{Conclusion and Future Work}
\label{sec:conc}
We have introduced and studied the out-of-order membership problem for regular
languages, and the out-of-order evaluation problem for monoids and semigroups.
We have shown that all these problems can be solved with constant time per
streamed symbol and linear space, and have aimed to characterize the space
complexity depending on the fixed monoid, semigroup, or language. Our results
give a complete classification of the space complexity of the problem for
monoids, between constant space for commutative monoids, logarithmic space for
$\mathbf{FL} \lor \mathbf{Com}$ monoids, and linear space otherwise. These
tractability boundaries happen to coincide with those of simultaneous
deterministic communication complexity 
from~\cite{tesson}, though as far as we can see this appears to be coincidental.
Our results further classify all constant-space semigroups, with
logarithmic lower bounds for all other cases. We have given partial results for
the other complexity regimes, but the precise characterizations remain elusive.
Refer again to \cref{all_table,all_ex_table} for a summary of our
results and examples.

Our work opens many directions for future work. The obvious question is to
classify
the remaining cases: in particular finding a superlogarithmic lower
bound for some semigroup with a sublinear space complexity, characterizing those
semigroups to which an $O(\sqrt{n})$ lower bound applies, or characterizing the
$O(1)$ regime for languages.
We could study the randomized complexity and
the non-deterministic complexity of our problems (see also~\cite{ada2010non}).
It would also be interesting to understand how the complexity changes for
problem variants, e.g.,
when some positions are not streamed and
implicitly filled with a default letter, or
when we are able to stream the same position multiple
times (with the same contents). This latter variant would invalidate some
algorithms (e.g., that of \cref{prp:aba})
so we believe it makes a difference in terms of complexity.
We could also study the setting where the word length is not known in advance
(though this would make the problem less symmetrical), or settings where the
positions are labeled with integers that are not necessarily positive, or
consecutive, or indeed with rational numbers.

We have also not discussed the complexity of determining the space complexity regime
when we are given a language as input (e.g., as an automaton); or the complexity of our algorithms when the
monoid or semigroup or language is not fixed but is also part of the input.
Another question, pertaining to the connection to property
testing~\cite{alon2001regular,bathie2025trichotomy}, would be to ask how the
space complexity is changed if we must only distinguish between words belonging
to the target language~$L$ and words that are sufficiently far from~$L$; we could also study strengthenings like computing the distance of the
word to~$L$.

Another interesting problem variant would be the \emph{earliest answer} setting,
where we would require the algorithm to give its answer about membership to~$L$
not at the end of the stream but as soon as it becomes certain, i.e., as soon as we reach
a partial word for which all completions belong to~$L$ or all completions do not
belong to~$L$. For some classes (e.g., commutative languages) it does not seem
difficult to enforce this, but we have not investigated whether this requirement would
increase the complexity of the problem in general.

One last question is of course about the complexity of out-of-order membership
for more general languages than regular languages, e.g., for context-free
languages we may not be able to have even an algorithm working in constant time
per streamed letter because it would in particular give a linear-time parsing
algorithm for arbitrary grammars (already in the special case where the letters
are streamed from left to right).

\bibliography{main}

\appendix

\end{document}